\documentclass{article}

     \PassOptionsToPackage{numbers, compress}{natbib}

\usepackage[final,nonatbib]{neurips_2022}

\makeatletter
\renewcommand{\@noticestring}{}
\makeatother

\usepackage{amssymb,amsmath,amsthm,amsfonts}
\usepackage{mathtools}
\usepackage{enumitem}
\usepackage[numbers,comma,sort&compress]{natbib}
\usepackage{graphicx}
\usepackage[font=small]{caption}
\usepackage[labelformat=simple]{subcaption}
\usepackage{float}
\usepackage[ruled,vlined,linesnumbered]{algorithm2e}
\usepackage{physics}
\usepackage{footnote}
\usepackage{xcolor}
\usepackage{bm}
\usepackage{multirow}
\usepackage{booktabs}

\usepackage{tikz}
\usetikzlibrary{arrows.meta}
\usetikzlibrary{positioning}

\usetikzlibrary{calc,patterns,angles,quotes}
\usepackage{siunitx}

\usepackage{hyperref}
\definecolor{mydarkblue}{rgb}{0,0.08,0.45}
\hypersetup{colorlinks=true, citecolor=mydarkblue,linkcolor=mydarkblue}

\setlist[enumerate]{itemsep=0pt}

\newtheorem{fact}{Fact}[section]

\usepackage{thmtools}
\usepackage{thm-restate}

\declaretheorem[name=Theorem,numberwithin=section]{theorem}

\declaretheorem[name=Lemma,numberwithin=section]{lemma}
\declaretheorem[name=Proposition,numberwithin=section]{proposition}
\declaretheorem[name=Corollary,numberwithin=section]{corollary}
\declaretheorem[name=Remark,numberwithin=section]{remark}

\numberwithin{equation}{section}

\newcommand{\eqn}[1]{(\ref{eqn:#1})}

\newcommand{\prb}[1]{(\ref{prb:#1})}
\newcommand{\thm}[1]{\hyperref[thm:#1]{Theorem~\ref*{thm:#1}}}
\newcommand{\cor}[1]{\hyperref[cor:#1]{Corollary~\ref*{cor:#1}}}
\newcommand{\defn}[1]{\hyperref[defn:#1]{Definition~\ref*{defn:#1}}}
\newcommand{\lem}[1]{\hyperref[lem:#1]{Lemma~\ref*{lem:#1}}}
\newcommand{\prop}[1]{Proposition~\ref{prop:#1}}
\newcommand{\fig}[1]{\hyperref[fig:#1]{Figure~\ref*{fig:#1}}}
\newcommand{\tab}[1]{Table~\ref{tab:#1}}
\newcommand{\algo}[1]{\hyperref[algo:#1]{Algorithm~\ref*{algo:#1}}}
\renewcommand{\sec}[1]{Section~\ref{sec:#1}}
\newcommand{\append}[1]{\hyperref[append:#1]{Appendix~\ref*{append:#1}}}
\newcommand{\fac}[1]{\hyperref[fac:#1]{Fact~\ref*{fac:#1}}}
\newcommand{\lin}[1]{\hyperref[lin:#1]{Line~\ref*{lin:#1}}}
\newcommand{\fnote}[1]{\hyperref[fnote:#1]{Footnote~\ref*{fnote:#1}}}

\SetKwInput{KwInput}{Input}
\SetKwInput{KwOutput}{Output}

\newcommand{\arxiv}[1]{\href{https://arxiv.org/abs/#1}{arXiv:#1}}

\renewcommand{\citet}{\cite}
\renewcommand{\citep}{\cite}

\let\originalleft\left
\let\originalright\right
\renewcommand{\left}{\mathopen{}\mathclose\bgroup\originalleft}
\renewcommand{\right}{\aftergroup\egroup\originalright}

\def\>{\rangle}
\def\<{\langle}

\newcommand{\R}{\mathbb{R}}
\newcommand{\C}{\mathbb{C}}

\newcommand{\E}{\mathbb{E}}

\def\P{{\mathcal P}}

\DeclareMathOperator{\poly}{poly}

\renewcommand{\d}{\mathrm{d}}

\newcommand{\range}[1]{[#1]}

\newcommand{\wt}[1]{\widetilde{#1}}

\let\oldnl\nl
\newcommand{\nonl}{\renewcommand{\nl}{\let\nl\oldnl}}

\makeatletter
\def\Ddots{\mathinner{\mkern1mu\raise\p@
\vbox{\kern7\p@\hbox{.}}\mkern2mu
\raise4\p@\hbox{.}\mkern2mu\raise7\p@\hbox{.}\mkern1mu}}
\makeatother
\definecolor{mygreen}{RGB}{80,180,0}

\title{Quantum Algorithms for Sampling Log-Concave Distributions and Estimating Normalizing Constants}

\author{%
Andrew M.\ Childs\\
Joint Center for Quantum Information and Computer Science,\\
Department of Computer Science, and\\
Institute for Advanced Computer Studies\\
University of Maryland \\
\texttt{amchilds@umd.edu}\\
\And
Tongyang Li\\
Center on Frontiers of Computing Studies and\\
School of Computer Science\\
Peking University\\
\texttt{tongyangli@pku.edu.cn}\\
\And
Jin-Peng Liu\\
Simons Institute and\\
Department of Mathematics\\
UC Berkeley\\
\texttt{jliu1219@terpmail.umd.edu}\\
\And
Chunhao Wang\\
Department of Computer Science and Engineering\\
Pennsylvania State University\\
\texttt{cwang@psu.edu}\\
\And
Ruizhe Zhang\\
Department of Computer Science\\
The University of Texas at Austin\\
\texttt{ruizhe@utexas.edu}
}

\begin{document}

\maketitle

\begin{abstract}
Given a convex function $f\colon\R^{d}\to\R$, the problem of sampling from a distribution $\propto e^{-f(x)}$ is called log-concave sampling. This task has wide applications in machine learning, physics, statistics, etc. In this work, we develop quantum algorithms for sampling log-concave distributions and for estimating their normalizing constants $\int_{\R^d}e^{-f(x)}\d x$. First, we use underdamped Langevin diffusion to develop quantum algorithms that match the query complexity (in terms of the condition number $\kappa$ and dimension $d$) of analogous classical algorithms that use gradient (first-order) queries, even though the quantum algorithms use only evaluation (zeroth-order) queries. For estimating normalizing constants, these algorithms also achieve quadratic speedup in the multiplicative error $\epsilon$. Second, we develop quantum Metropolis-adjusted Langevin algorithms with query complexity $\widetilde{O}(\kappa^{1/2}d)$ and $\widetilde{O}(\kappa^{1/2}d^{3/2}/\epsilon)$ for log-concave sampling and normalizing constant estimation, respectively, achieving polynomial speedups in $\kappa,d,\epsilon$ over the best known classical algorithms
by exploiting quantum analogs of the Monte Carlo method and quantum walks. We also prove a $1/\epsilon^{1-o(1)}$ quantum lower bound for estimating normalizing constants, implying near-optimality of our quantum algorithms in $\epsilon$.
\end{abstract}

\section{Introduction}\label{sec:intro}
Sampling from a given distribution is a fundamental computational problem.
For example, in statistics, samples can determine confidence intervals or explore posterior distributions. In machine learning, samples are used for regression and to train supervised learning models. In optimization, samples from well-chosen distributions can produce points near local or even global optima.

Sampling can be nontrivial even when the distribution is known. Indeed, efficient sampling is often a challenging computational problem, and bottlenecks the running time in many applications. Many efforts have been made to develop fast sampling methods. Among those, one of the most successful tools is Markov Chain Monte Carlo (MCMC), which uses a Markov chain that converges to the desired distribution to (approximately) sample from it.

Here we focus on the fundamental task of \emph{log-concave sampling}, i.e., sampling from a distribution proportional to $e^{-f}$ where $f\colon\R^{d}\to\R$ is a convex function. This covers many practical applications such as multivariate Gaussian distributions and exponential distributions. Provable performance guarantees for log-concave sampling have been widely studied~\citep{DCWY18}. A closely related problem is estimating the normalizing constants of log-concave distributions, which also has many applications~\citep{ge2020estimating}.

Quantum computing has been applied to speed up many classical algorithms based on Markov processes, so it is natural to investigate quantum algorithms for log-concave sampling. If we can prepare a quantum state whose amplitudes are the square roots of the corresponding probabilities, then measurement yields a random sample from the desired distribution. In this approach, the number of required qubits is only poly-logarithmic in the size of the sample space. Unfortunately, such a quantum state probably cannot be efficiently prepared in general, since this would imply $\mathsf{SZK} \subseteq \mathsf{BQP}$~\citep{aharonov2003adiabatic}. Nevertheless, in some cases, quantum algorithms can achieve polynomial speedup over classical algorithms. Examples include uniform sampling on a 2D lattice~\citep{Richter2007}, %
estimating partition functions~\citep{wocjan2008speedup,wocjan2009quantum,montanaro2015quantum,harrow2019adaptive,AHN21}, and estimating volumes of convex bodies~\citep{cch19}. However, despite the importance of sampling log-concave distributions and estimating normalizing constants, we are not aware of any previous quantum speedups for general instances of these problems.

\paragraph{Formulation}
In this paper, we consider a $d$-dimensional convex function $f\colon\R^{d}\to\R$ which is $L$-smooth and $\mu$-strongly convex, i.e., $\mu,L>0$ and for any $x,y\in\R^{d}$, $x\neq y$,
\begin{align}\label{eqn:kappa-defn}
\frac{f(y)-f(x)-\langle\nabla f(x),y-x\rangle}{\|x-y\|_{2}^{2}/2}\in[\mu,L].
\end{align}
We denote by $\kappa:=L/\mu$ the condition number of $f$. The corresponding log-concave distribution has probability density $\rho_{f}\colon\R^{d}\to\R$ with
\begin{align}\label{prb:log-concave-distribution}
\rho_{f}(x):=\frac{e^{-f(x)}}{Z_f},
\end{align}
where the normalizing constant is
\begin{align}\label{prb:log-concave}
Z_{f}:=\int_{x\in \R^d}e^{-f(x)}\,\d x.
\end{align}
When there is no ambiguity, we abbreviate $\rho_{f}$ and $Z_{f}$ as $\rho$ and $Z$, respectively. Given an $\epsilon\in(0,1)$,
\begin{itemize}[leftmargin=*]
\item the goal of \emph{log-concave sampling} is to output a random variable with distribution $\widetilde{\rho}$ such that $\|\widetilde{\rho}-\rho\|\leq \epsilon$, and
\item the goal of \emph{normalizing constant estimation} is to output a value $\widetilde{Z}$ such that with probability at least $2/3$, $(1-\epsilon)Z \leq \widetilde{Z} \leq(1+\epsilon)Z$.
\end{itemize}
Here $\|\cdot\|$ is a certain norm. We consider the general setting where the function $f$ is specified by an oracle. In particular, we consider the quantum evaluation oracle $O_{f}$, a standard model in the quantum computing literature~\citep{cch19,vanApeldoorn2018convex,chakrabarti2018quantum,zll21}.
The evaluation oracle acts as
\begin{align}\label{eqn:oracle-evaluation}
O_{f}|x,y\>=|x,f(x)+y\>\quad\ \forall x\in\R^{d}, y\in\R.
\end{align}
(Quantum computing notations are briefly explained in \sec{notations}.)
We also consider the quantum gradient oracle $O_{\nabla f}$ with
\begin{align}\label{eqn:oracle-gradient}
O_{\nabla f}|x,z\>=|x,\nabla f(x)+z\>\quad\ \forall x,z\in\R^{d}.
\end{align}
In other words, we allow superpositions of queries to both function evaluations and gradients. The essence of quantum speedup is the ability to compute with carefully designed superpositions.

\paragraph{Contributions}%
Our main results are quantum algorithms that speed up log-concave sampling and normalizing constant estimation.

\begin{theorem}[Main log-concave sampling result]\label{thm:main-log-concave}
Let $\rho$ denote the log-concave distribution (\ref{prb:log-concave-distribution}). There exist quantum algorithms that output a random variable distributed according to $\widetilde{\rho}$ such that%
\begin{itemize}[leftmargin=*,nosep]
  \item $W_2(\widetilde{\rho}, \rho)\leq\epsilon$ where $W_{2}$ is the Wasserstein 2-norm~\eqn{W2-defn}, using $\widetilde O(\kappa^{7/6}d^{1/6}\epsilon^{-1/3}+\kappa d^{1/3}\epsilon^{-2/3})$ queries to the quantum evaluation oracle \eqn{oracle-evaluation}; or
\item $\|\widetilde{\rho}-\rho\|_{\text{TV}}\leq\epsilon$ where $\|\cdot\|_{\text{TV}}$ is the total variation distance~\eqn{TV}, using $\widetilde O\big(\kappa^{1/2} d\big)$ queries to the quantum gradient oracle \eqn{oracle-gradient}, or $\widetilde O\big(\kappa^{1/2} d^{1/4}\big)$ queries when the initial distribution is warm (formally defined in \append{classical_MALA}).
\end{itemize}
\end{theorem}
In the above results, the query complexity $\widetilde O(\kappa^{7/6}d^{1/6}\epsilon^{-1/3}+\kappa d^{1/3}\epsilon^{-2/3})$ is achieved by our quantum ULD-RMM algorithm. Although the quantum query complexity is the same as the best known classical result~\cite{sl19}, we emphasize that our quantum algorithm uses a zeroth-order oracle while \cite{sl19} uses a first-order oracle. The query complexity $\widetilde O\big(\kappa^{1/2} d\big)$ is achieved by our quantum MALA algorithm that uses a first-order oracle (as in classical algorithms). This is a quadratic speedup in $\kappa$ compared with the best known classical algorithm~\cite{lst20}. With a warm start, our quantum speedup is even more significant: we achieve quadratic speedups in $\kappa$ and $d$ as compared with the best known classical algorithm with a warm start~\cite{WSC21}.

\begin{theorem}[Main normalizing constant estimation result]\label{thm:main-normalizing}
There exist quantum algorithms that estimate the normalizing constant by $\widetilde{Z}$ within multiplicative error $\epsilon$ with probability at least $3/4$,
\begin{itemize}[leftmargin=*,nosep]
  \item using $\widetilde O(\kappa^{7/6}d^{7/6}\epsilon^{-1}+\kappa d^{4/3}\epsilon^{-1})$ queries to the quantum evaluation oracle \eqn{oracle-evaluation}; or
\item using $\widetilde O(\kappa^{1/2}d^{3/2}\epsilon^{-1})$ queries to the quantum gradient oracle \eqn{oracle-gradient}.%
\end{itemize}
Furthermore, this task has quantum query complexity at least $\Omega(\epsilon^{-1+o(1)})$ (\thm{lower-bound}).
\end{theorem}
Our query complexity $\widetilde O(\kappa^{7/6}d^{7/6}\epsilon^{-1}+\kappa d^{4/3}\epsilon^{-1})$ for normalizing constant estimation achieves a quadratic speedup in precision compared with the best known classical algorithm~\cite{ge2020estimating}. More remarkably, our quantum ULD-RMM algorithm again uses a zeroth-order oracle while the slower best known classical algorithm uses a first-order oracle~\cite{ge2020estimating}. Our quantum algorithm working with a first-order oracle achieves polynomial speedups in all parameters compared with the best known classical algorithm~\cite{ge2020estimating}. Moreover, the precision-dependence of our quantum algorithms is nearly optimal, which is quadratically better than the classical lower bound in $1/\epsilon$~\cite{ge2020estimating}.

To the best of our knowledge, these are the first quantum algorithms with quantum speedup for the fundamental problems of log-concave sampling and estimating normalizing constants. We explore multiple classical techniques %
including the underdamped Langevin diffusion (ULD) method~\cite{dm17,dk19,dmm19,vw19}, the randomized midpoint method for underdamped Langevin diffusion (ULD-RMM)~\cite{sl19,roy2019stochastic}, and the Metropolis adjusted Langevin algorithm (MALA)~\cite{DCWY18,chen2019fast,lst20,cla21,WSC21,lee2021lower}, and achieve quantum speedups. Our main contributions are as follows.
\begin{itemize}[leftmargin=*]
  \item \emph{Log-concave sampling.} For this problem, our quantum algorithms based on ULD and ULD-RMM have the same query complexity as the best known classical algorithms, but our quantum algorithms only use a zeroth-order (evaluation) oracle, while the classical algorithms use the first-order (gradient) oracle. For MALA, this improvement on the order of oracles is nontrivial, but we can use the quantum gradient oracle in our quantum MALA algorithm to achieve a quadratic speedup %
      in the condition number $\kappa$. Furthermore, given a warm-start distribution, our quantum algorithm achieves a quadratic speedup in all parameters.

  \item \emph{Normalizing constant estimation.} For this problem, our quantum algorithms provide larger speedups. In particular, our quantum algorithms based on ULD and ULD-RMM achieve quadratic speedup in the multiplicative precision $\epsilon$ (while using a zeroth-order oracle) compared with the corresponding best-known classical algorithms (using a first-order oracle). Our quantum algorithm based on MALA achieves polynomial speedups in all parameters. %
  Furthermore, we prove that our quantum algorithm is nearly optimal in terms of $\epsilon$.
\end{itemize}

We summarize our results and compare them to previous classical algorithms in \tab{alg-sampling} and \tab{alg-estimating}. See \append{related-work} for more detailed comparisons to related classical and quantum work.

\begin{table}[ht]
\caption{Summary of the query complexities of classical and quantum algorithms for sampling a $d$-dimensional log-concave distribution. Here $\kappa=L/\mu$ in \eqn{kappa-defn} and $\epsilon$ is the error in the designated norm.}
\label{tab:alg-sampling}
\vskip -0.2in
\begin{center}
\footnotesize
\begin{tabular}{cccc}
\toprule
Method & Oracle & Complexity & Norm \\
\midrule
\hspace{-1mm}ULD~\citep{CCBJ18}    & gradient & $\widetilde O\left(\kappa^2d^{1/2}\epsilon^{-1}\right)$ & $W_{2}$ \\
\hspace{-1mm}ULD-RMM~\citep{sl19} & gradient & $\widetilde O\left(\kappa^{7/6}d^{1/6}\epsilon^{-1/3}+\kappa d^{1/3}\epsilon^{-2/3}\right)$ & $W_{2}$ \\
\hspace{-1mm}MALA~\citep{lst20}    & gradient & $\widetilde O\left(\kappa d\right)$ & TV \\
\hspace{-1mm}MALA with warm start~\citep{WSC21}    & gradient & $\widetilde O\left(\kappa d^{1/2}\right)$ & TV       \\
\midrule
\hspace{-1mm}Quantum Inexact ULD~(\thm{quantum-IULD})    & \hspace{-1mm}evaluation\hspace{-1mm} & $\widetilde O\left(\kappa^2d^{1/2}\epsilon^{-1}\right)$ & $W_{2}$ \\
\hspace{-1mm}Quantum Inexact ULD-RMM~(\thm{quantum-IULD-RMM})\hspace{-2mm}   & \hspace{-1mm}evaluation\hspace{-1mm} & $\widetilde O\left(\kappa^{7/6}d^{1/6}\epsilon^{-1/3}+\kappa d^{1/3}\epsilon^{-2/3}\right)$ & $W_{2}$ \\
\hspace{-1mm}Quantum MALA~(\thm{q_MALA_sampling})     & gradient & $\widetilde O\left(\kappa^{1/2} d\right)$ & TV        \\
\hspace{-1mm}Quantum MALA (warm start) (\thm{q_MALA_warm})\hspace{-2mm}   & gradient & $\widetilde O\left(\kappa^{1/2} d^{1/4}\right)$ & TV \\
\bottomrule
\end{tabular}
\end{center}
\vspace{-3mm}
\end{table}

\begin{table}[ht]
\caption{Summary of the query complexities of classical and quantum algorithms for estimating the normalizing constant of a $d$-dimensional log-concave distribution. Here $\kappa=L/\mu$ in \eqn{kappa-defn} and $\epsilon$ is the multiplicative error.}
\label{tab:alg-estimating}
\vskip -0.2in
\begin{center}
\footnotesize
\begin{tabular}{cccr}
\toprule
Method & Oracle & Complexity \\
\midrule
\hspace{-2mm}Multilevel ULD~\citep{ge2020estimating}    & gradient & $\widetilde O\left(\kappa^2d^{3/2}\epsilon^{-2}\right)$ \\
\hspace{-2mm}Multilevel ULD-RMM~\citep{ge2020estimating} & gradient & $\widetilde O\left(\kappa^{7/6}d^{7/6}\epsilon^{-2}+\kappa d^{4/3}\epsilon^{-2}\right)$ \\
\hspace{-2mm}MALA~\citep{ge2020estimating}    & gradient & $\widetilde O\left(\kappa d^2\epsilon^{-2}\max\{1,\frac{\kappa}{d}\}\right)$ \\
\midrule
\hspace{-2mm}Multilevel Quantum Inexact ULD~(\thm{quantum-annealing-multilevel-ULD})    & \hspace{-2mm}evaluation\hspace{-1mm} & $\widetilde O\left(\kappa^2d^{3/2}\epsilon^{-1}\right)$ \\
\hspace{-2mm}Multilevel Quantum Inexact ULD-RMM~(\thm{quantum-annealing-multilevel-ULD-RMM})\hspace{-2mm}   & \hspace{-1mm}evaluation\hspace{-1mm} & $\widetilde O\left(\kappa^{7/6}d^{7/6}\epsilon^{-1}+\kappa d^{4/3}\epsilon^{-1}\right)$ \\
\hspace{-2mm}Quantum annealing with Quantum MALA~(\thm{quantum-MALA-walk})     & gradient & $\widetilde O\left(\kappa^{1/2}d^{3/2}\epsilon^{-1}\right)$     \\
\bottomrule
\end{tabular}
\end{center}
\vskip -0.1in
\end{table}

\paragraph{Techniques}

In this work, we develop a systematic approach for studying the complexity of quantum walk mixing and show that for any reversible classical Markov chain, we can obtain quadratic speedup for the mixing time as long as the initial distribution is warm. In particular, we apply the quantum walk and quantum annealing in the context of Langevin dynamics and achieve polynomial quantum speedups.

The technical ingredients of our quantum algorithms are highlighted below.
\begin{itemize}[leftmargin=*]

  \item \emph{Quantum simulated annealing (\lem{informal_slowly_varying}).} Our quantum algorithm for estimating normalizing constants combines the quantum simulated annealing framework of~\citet{wocjan2008speedup} and the quantum mean estimation algorithm of~\citet{montanaro2015quantum}. For each type of Langevin dynamics (which are random walks), we build a corresponding quantum walk. %
  Crucially, the spectral gap of the random walk is quadratically amplified in the phase gap of the corresponding quantum walk. This allows us to use a Grover-like procedure to produce the stationary distribution state given a sufficiently good initial state. In the simulated annealing framework, this initial state is the stationary distribution state of the previous Markov chain. %

  \item \emph{Effective spectral gap (\lem{effect_spectral_gap}).} We show how to leverage a ``warm'' initial distribution to achieve a quantum speedup for sampling. Classically, a warm start leads to faster mixing even if the spectral gap is small. Quantumly, we generalize the notion of ``effective spectral gap''~\citep{reichardt2009span,lee2011quantum,cch19} to our more general sampling problem. We show that with a bounded warmness parameter, quantum algorithms can achieve a quadratic speedup in the mixing time. By viewing the sampling problem as a simulated annealing process with only one Markov chain, we prove a quadratic speedup for quantum MALA by analyzing the effective spectral gap.
  \item \emph{Quantum gradient estimation (\lem{SmoothQuantumGradient}).} We adapt Jordan's quantum gradient algorithm~\citep{Jor05} to the ULD and ULD-RMM algorithms and give rigorous proofs to bound the sampling error due to gradient estimation errors.
\end{itemize}

\paragraph{Open questions}%
Our work raises several natural questions for future investigation:
\begin{itemize}[leftmargin=*]
\item Can we achieve quantum speedup in $d$ and $\kappa$ for unadjusted Langevin algorithms such as ULD and ULD-RMM? The main difficulty is that ULD and ULD-RMM are irreversible, while most available quantum walk techniques only apply to reversible Markov chains. New techniques might be required to resolve this question.  %
\item Can we achieve further quantum speedup for estimating normalizing constants with a warm start distribution? This might require a more refined version of quantum mean estimation. %
\item Can we give quantum algorithms for estimating normalizing constants with query complexity sublinear in $d$? Such a result would give a provable quantum-classical separation due to the $\Omega(d^{1-o(1)}/\epsilon^{2-o(1)})$ classical lower bound  proved in \citet{ge2020estimating}.
\end{itemize}

\section{Preliminaries}\label{sec:notations}

\paragraph{Basic definitions of quantum computation}
Quantum mechanics is formulated in terms of linear algebra. The \emph{computational basis} of $\C^{d}$ is $\{\vec{e}_{0},\ldots,\vec{e}_{d-1}\}$, where $\vec{e}_{i}=(0,\ldots,1,\ldots,0)^{\top}$ with the $1$ in the $(i+1)^{\text{st}}$ position. We use \emph{Dirac notation}, writing $|i\>$ (called a ``ket") for $\vec{e}_{i}$ and $\<i|$ (a ``bra") for $\vec{e}_{i}^{\top}$.

The \emph{tensor product} of quantum states is their Kronecker product: if $|u\>\in\C^{d_{1}}$ and $|v\>\in\C^{d_{2}}$, then we have $|u\>\otimes|v\>\in\C^{d_{1}}\otimes\C^{d_{2}}$ with
\begin{align}
|u\>\otimes|v\>=(u_{0}v_{0},u_{0}v_{1},\ldots,u_{d_{1}-1}v_{d_{2}-1})^{\top}.
\end{align}
The basic element of quantum information is a \emph{qubit}, a quantum state in $\C^{2}$, which can be written as $a|0\>+b|1\>$ for some $a,b\in\C$ with $|a|^{2}+|b|^{2}=1$. An $n$-qubit tensor product state can be written as $|v_{1}\>\otimes\cdots\otimes|v_{n}\>\in\C^{2^{n}}$, where for any $i\in\range{n}$, $|v_{i}\>$ is a one-qubit state. Note however that most states in $\C^{2^n}$ are not product states. We sometimes abbreviate $|u\>\otimes|v\>$ as $|u\>|v\>$.

Operations on quantum states are \emph{unitary transformations}. They are typically stated in the circuit model, where a \emph{$k$-qubit gate} is a unitary matrix in $\C^{2^{k}}$. Two-qubit gates are \emph{universal}, i.e., every $n$-qubit gate can be decomposed into a product of gates that act as the identity on $n-2$ qubits and as some two-qubit gate on the other $2$ qubits. The \emph{gate complexity} of an operation refers to the number of two-qubit gates used in a quantum circuit for realizing it.

Quantum access to a function, referred to as a \emph{quantum oracle}, must be reversible and allow access to different values of the function in \emph{superposition} (i.e., for linear combinations of computational basis states). For example, consider the unitary evaluation oracle $O_f$ defined in \eqn{oracle-evaluation}. Given a probability distribution $\{p_{i}\}_{i=1}^{n}$
and a set of points $\{x_{i}\}_{i=1}^{n}$, we have
\begin{align}\label{eq:def_eva_oracle}
O_{f}\sum_{i=1}^{n}\sqrt{p_{i}}|x_{i}\>|0\>=\sum_{i=1}^{n}\sqrt{p_{i}}|x_{i}\>|f(x_{i})\>.
\end{align}
Then a measurement would give $f(x_i)$ with probability $p_{i}$. However, a quantum oracle can not only simulate random sampling, but can enable uniquely quantum behavior through interference. Examples include amplitude amplification---the main idea behind Grover's search algorithm~\cite{grover1996fast} and the amplitude estimation procedure used in this paper---and many other quantum algorithms relying on coherent quantum access to a function. %
Similar arguments apply to the quantum gradient oracle \eqn{oracle-gradient}. If a classical oracle can be computed by an explicit classical circuit, then the corresponding quantum oracle can be implemented by a quantum circuit of approximately the same size. Therefore, these quantum oracles provide a useful framework for understanding the quantum complexity of log-concave sampling and normalizing constant estimation.

To sample from a distribution $\pi$, it suffices to prepare the state $\ket{\pi}:=\sum_x\sqrt{\pi_x}\ket{x}$ and then measure it. For a Markov chain specified by a transition matrix $P$ with stationary distribution $\pi$, one can construct a corresponding \emph{quantum walk} operator $W(P)$.
Intuitively, quantum walks can be viewed as applying a sequence of quantum unitaries on a quantum state encoding the initial distribution
to rotate it to the subspace of stationary distribution $\ket{\pi}$. The number of rotations needed (i.e., the angle between the initial distribution and stationary distribution) depends on the spectral gap of $P$, and a quantum algorithm can achieve a quadratic speedup via \emph{quantum phase estimation} and \emph{amplification} algorithms.
More background on quantum walk is given in \append{quantum_speedup_MALA}.

\paragraph{Notations}
Throughout the paper, the big-O notations $O(\cdot)$, $o(\cdot)$, $\Omega(\cdot)$, and $\Theta(\cdot)$ follow common definitions.
The $\tilde{O}$ notation omits poly-logarithmic terms, i.e., $\tilde{O}(f):=O(f\mathrm{poly}(\log f))$.
We say a function $f$ is \emph{$L$-Lipschitz continuous} at $x$ if $|f(x)-f(y)|\leq L\|x-y\|$ for all $y$ sufficiently near $x$.
The total variation distance (TV-distance) between two functions $f,g\colon\R^{d}\to\R$ is defined as
\begin{align}\label{eqn:TV}
\|f-g\|_{\text{TV}}:=\int_{\R^{d}}|f(x)-g(x)|\,\d x.
\end{align}

Let $\mathcal{B}(\R^{d})$ denote the Borel $\sigma$-field of $\R^{d}$. Given probability measures $\mu$ and $\nu$ on $(\R^{d},\mathcal{B}(\R^{d}))$, a \emph{transference plan} $\zeta$ between $\mu$ and $\nu$ is defined as a probability measure on $(\R^{d}\times \R^{d},\mathcal{B}(\R^{d})\times \mathcal{B}(\R^{d}))$ such that for any $A\subseteq\R^{d}$, $\zeta(A\times\R^{d})=\mu(A)$ and $\zeta(\R^{d}\times A)=\nu(A)$. We let $\Gamma(\mu,\nu)$ denote the set of all transference plans. We let
\begin{align}\label{eqn:W2-defn}
\hspace{-3mm}W_{2}(\mu,\nu):=\left(\inf_{\zeta\in\Gamma(\mu,\nu)}\int_{\R^{d}\times\R^{d}}\|x-y\|_{2}^{2}\,\d\zeta(x,y)\right)^{\frac{1}{2}}
\end{align}
denote the Wasserstein 2-norm between $\mu$ and $\nu$.

\section{Quantum Algorithm for Log-Concave Sampling}
In this section, we describe several quantum algorithms for sampling log-concave distributions.

\paragraph{Quantum inexact ULD and ULD-RMM} We first show that the gradient oracle in the classical ULD and ULD-RMM algorithms can be efficiently simulated by the quantum evaluation oracle via quantum gradient estimation.
Suppose we are given access to the evaluation oracle \eqn{oracle-evaluation} for $f(x)$. Then by Jordan's algorithm \cite{Jor05} (see \lem{SmoothQuantumGradient} for details), there is a quantum algorithm that can compute $\nabla f(x)$ with a polynomially small $\ell_1$-error by querying the evaluation oracle $O(1)$ times.
Using this, we can prove the following theorem (see \append{q_ULD_sampling} for details).

\begin{theorem}[Informal version of \thm{quantum-IULD} and \thm{quantum-IULD-RMM}]\label{thm:informal_inexact}
  Let $\rho\propto e^{-f}$ be a $d$-dimensional log-concave distribution with $f$ satisfying \eqn{kappa-defn}. Given a quantum evaluation oracle for $f$,
  \begin{itemize}[leftmargin=*,nosep]
    \item the quantum inexact ULD algorithm uses $\widetilde{O}(\kappa^2d^{1/2}\epsilon^{-1})$ queries, and
    \item the quantum inexact ULD-RMM algorithm uses $\widetilde{O}(\kappa^{7/6}d^{1/6}\epsilon^{-1/3}+\kappa d^{1/3}\epsilon^{-2/3})$ queries,
  \end{itemize}
  to quantumly sample from a distribution that is $\epsilon$-close to $\rho$ in $W_2$-distance.
\end{theorem}

We note that the query complexities of our quantum algorithms using a \emph{zeroth-order} oracle match the state-of-the-art classical ULD \citep{CCBJ18} and ULD-RMM \citep{sl19} complexities with a \emph{first-order} oracle.
The main technical difficulty of applying the quantum gradient algorithm is that it produces a \emph{stochastic gradient oracle} in which the output of the quantum algorithm $\mathbf{g}$ satisfies $\|\E[\mathbf{g}]-\nabla f(x)\|_1\leq d^{-\Omega(1)}$. In particular, the randomness of the gradient computation is ``entangled'' with the randomness of the Markov chain. %
We use the classical analysis of ULD and ULD-RMM processes \citep{roy2019stochastic} to prove that the stochastic gradient will not significantly slow down the mixing of ULD processes, and that the error caused by the quantum gradient algorithm can be controlled.

\paragraph{Quantum MALA} We next propose two quantum algorithms with lower query complexity than classical MALA, one with a Gaussian initial distribution and another with a warm-start distribution. The main technical tool we use is a quantum walk in continuous space. %

The classical MALA (i.e., Metropolized HMC) starts from a Gaussian distribution ${\cal N}(0, L^{-1}I_d)$ and performs a leapfrog step in each iteration. It is well-known that the initial Gaussian state
\begin{align}
  \ket{\rho_0} = \int_{\R^d} \left(\frac{L}{2\pi}\right)^{d/4} e^{-\frac{L}{4}\|z-x^\star\|_2^2}\ket{z} \, \d z
\end{align}
can be efficiently prepared.
We show that the quantum walk update operator
\begin{align}
  U:=\int_{\R^d}\d x\int_{\R^d}\d y \, \sqrt{p_{x\rightarrow y}} \ket{x}\bra{x}\otimes \ket{y}\bra{0}
\end{align}
can be efficiently implemented, where $p_{x\rightarrow y}:=p(x,y)$ is the transition density from $x$ to $y$, and the density $p$ satisfies $\int_{\R^d} p(x,y)\,\d y=1$ for any $x\in \R^d$.

\begin{lemma}[Informal version of \lem{continuous_walk}]\label{lem:informal_implement}
  The continuous-space quantum walk operator corresponding to the MALA Markov chain can be implemented with $O(1)$ gradient and evaluation queries.
\end{lemma}

In general, it is difficult to quantumly speed up the mixing time of a classical Markov chain, which is upper bounded by $O(\delta^{-1}\log(\rho_{\min}^{-1}))$,
 where $\delta$ is the spectral gap. However, \citet{wocjan2008speedup} shows that a quadratic speedup is possible when following a sequence of \emph{slowly-varying} Markov chains. More specifically, let $\rho_0, \dots , \rho_r$ be the stationary distributions of the \emph{reversible} Markov chains ${\cal M}_0,\dots, {\cal M}_r$ and let $\ket{\rho_0},\dots, \ket{\rho_r}$ be the corresponding quantum states. Suppose $|\bra{\rho_i}\rho_{i+1}\rangle|\geq p$ for all $i\in \{0, \dots, r - 1\}$, and suppose the spectral gaps of ${\cal M}_0,\dots, {\cal M}_r$ are lower-bounded by $\delta$. Then we can prepare a quantum state $\ket{\widetilde{\rho_r}}$ that is $\epsilon$-close to $\ket{\rho_r}$ using $\tilde{O}\left(\delta^{-1/2}rp^{-1}\right)$ quantum walk steps. To fulfill the slowly-varying condition, we consider an annealing process that goes from $\rho_0={\cal N}(0, L^{-1}I_d)$ to the target distribution $\rho_{M+1}=\rho$ in $M=\widetilde{O}(\sqrt{d})$ stages. At the $i$th stage, the stationary distribution is $\rho_i\propto e^{-f_i}$ with $f_i:=f+\frac{1}{2}\sigma_i^{-2}\|x\|^2$. By properly choosing $\sigma_1\leq \cdots\leq  \sigma_M$, we prove that this sequence of Markov chains is slowly varying.
\begin{lemma}[Informal version of \lem{slowly_varying}]\label{lem:informal_slowly_varying}
  If we take $\sigma_1^2=\frac{\epsilon}{2dL}$ and $\sigma^2_{i+1}=(1+\frac{1}{\sqrt{d}})\sigma_i^2$, then for $0\leq i\leq M$, we have $|\langle \rho_i|\rho_{i+1}\rangle|\geq \Omega(1)$.
\end{lemma}

Combining \lem{informal_implement}, \lem{informal_slowly_varying}, and results on the mixing time of MALA \citep{lst20}, we have:
\begin{theorem}[Informal version of \thm{q_MALA_sampling}]\label{thm:informal_MALA_1}
  Let $\rho\propto e^{-f}$ be a $d$-dimensional log-concave distribution with $f$ satisfying~\eqref{eqn:kappa-defn}. There is a quantum algorithm (\algo{informal_MALA}) that prepares a state $\ket{\widetilde{\rho}}$ with $\|\ket{\widetilde{\rho}}-\ket{\rho}\|\leq \epsilon$ using
  $
    \widetilde{O}(\kappa^{1/2}d)
  $
  gradient and evaluation oracle queries.
\end{theorem}

\begin{algorithm}[htbp]
  \KwInput{Evaluation oracle ${\cal O}_f$, gradient oracle ${\cal O}_{\nabla f}$, smoothness parameter $L$, convexity parameter $\mu$}
  \KwOutput{Quantum state $\ket{\widetilde{\rho}}$ close to the stationary distribution state $\int_{\R^d} e^{-f(x)/2}\,\d \ket{x}$}
  Compute the cooling schedule parameters $\sigma_1,\dots, \sigma_M$\\
  Prepare the state $\ket{\rho_0}\propto \int_{\R^d} e^{-\frac{1}{4}\|x\|^2/\sigma_1^2}\,\d \ket{x}$\\
  \For{$i\gets 1,\dots, M$}{
    Construct ${\cal O}_{f_i}$ and ${\cal O}_{\nabla f_i}$ where $f_i(x)=f(x)+\frac{1}{2}\|x\|^2/\sigma_i^2$\\
    Construct quantum walk update unitary $U$ with ${\cal O}_{f_i}$ and ${\cal O}_{\nabla f_i}$\\
    Implement the quantum walk operator and the approximate reflection $\widetilde{R}_i$\\
    Prepare $\ket{\rho_i}$ by performing $\frac{\pi}{3}$-amplitude amplification with $\widetilde{R}_i$ on the state $\ket{\rho_{i-1}}\ket{0}$
  }
  \Return $\ket{\rho_M}$
  \caption{\textsc{QuantumMALAforLog-concaveSampling} (Informal)}
  \label{algo:informal_MALA}
\end{algorithm}

For the classical MALA with a Gaussian initial distribution, it was shown by \citet{lee2021lower} that the mixing time is at least $\widetilde{\Omega}(\kappa d)$. \thm{informal_MALA_1} quadratically reduces the $\kappa$ dependence.

Note that \algo{informal_MALA} uses a first-order oracle, instead of the zeroth-order oracle used in the quantum ULD algorithms. The technical barrier to applying the quantum gradient algorithm (\lem{SmoothQuantumGradient}) in the quantum MALA is to analyze the classical MALA with a stochastic gradient oracle. We currently do not know whether the ``entangled randomness'' dramatically increases the mixing time.

More technical details and proofs are provided in \append{log-concave_sampling}.

\section{Quantum Algorithm for Estimating Normalizing Constants}

In this section, we apply our quantum log-concave sampling algorithms to the normalizing constant estimation problem. A very natural approach to this problem is via MCMC, which constructs a multi-stage annealing process and uses a sampler at each stage to solve a mean estimation problem. We show how to quantumly speed up these annealing processes and improve the query complexity of estimating normalizing constants.

\paragraph{Quantum speedup for the standard annealing process}
We first consider the standard annealing process for log-concave distributions, as already applied in the previous section. Recall that we pick parameters $\sigma_1<\sigma_2<\cdots <\sigma_M$ and construct a sequence of Markov chains with stationary distributions $\rho_i\propto e^{-f_i}$, where $f_i=f+\frac{1}{2\sigma_i^2}\|x\|^2$. Then, at the $i$th stage, we estimate the expectation
\begin{align}
  \E_{\rho_i}[g_i]~~\text{where}~~g_i=\exp\left( \frac{1}{2}(\sigma_i^{-2}-\sigma_{i+1}^{-2})\|x\|^2 \right).
\end{align}
If we can estimate each expectation with relative error at most $O(\epsilon/M)$, then the product of these $M$ quantities estimates the normalizing constant $Z=\int_{\R^d}e^{-f(x)}\,\d x$ with relative error at most $\epsilon$.

For the mean estimation problem, \citet{montanaro2015quantum} showed that when the relative variance $\frac{\mathbf{Var}_{\rho_i}[g_i]}{\E_{\rho_i}[g_i]^2}$ is constant, there is a quantum algorithm for estimating the expectation $\E_{\rho_i}[g_i]$ within relative error at most $\epsilon$ using $\widetilde{O}(1/\epsilon)$ quantum samples from the distribution $\rho_i$. Our annealing schedule satisfies the bounded relative variance condition. Therefore, by the quantum mean estimation algorithm, we improve the sampling complexity of the standard annealing process from $\widetilde{O}(M^2\epsilon^{-2})$ to $\widetilde{O}(M\epsilon^{-1})$.

To further improve the query complexity, we consider using the quantum MALAs developed in the previous section to generate samples. Observe that \algo{informal_MALA} outputs a quantum state corresponding to some distribution that is close to $\rho_i$, instead of an individual sample. If we can estimate the expectation without destroying the quantum state, then we can reuse the state and evolve it for the $(i+1)$st Markov chain. Fortunately, we can use non-destructive mean estimation to estimate the expectation and restore the initial states. A detailed error analysis of this algorithm can be found in \citet{cch19,harrow2019adaptive}. We first prepare $\widetilde{O}(M\epsilon^{-1})$ copies of initial states corresponding to the Gaussian distribution ${\cal N}(0, L^{-1}I_d)$. Then, for each stage, we apply the non-destructive mean estimation algorithm to estimate the expectation $\E_{\rho_i}[g_i]$ and then run quantum MALA to evolve the states $\ket{\rho_i}$ to $\ket{\rho_{i+1}}$. This gives our first quantum algorithm for estimating normalizing constants.

\begin{theorem}[Informal version of \thm{quantum-MALA-walk}]
  Let $Z$ be the normalizing constant in \prb{log-concave}. There is a quantum algorithm (\algo{informal_MALA_est}) that outputs an estimate $\widetilde Z$ with relative error at most $\epsilon$
  using $\widetilde O\left(d^{3/2}\kappa^{1/2}\epsilon^{-1}\right)$
  queries to the quantum gradient and evaluation oracles.
\end{theorem}

\begin{algorithm}[htbp]
  \KwInput{Evaluation oracle ${\cal O}_f$, gradient oracle ${\cal O}_{\nabla f}$}
  \KwOutput{Estimate $\widetilde{Z}$ of $Z$ with relative error at most $\epsilon$}
  $M\gets \widetilde{O}(\sqrt{d})$, $K\gets \widetilde{O}(\epsilon^{-1})$\\
  Compute the cooling schedule parameters $\sigma_1,\dots, \sigma_M$\\
  \For{$j\gets 1,\dots, K$}{
    Prepare the state $\ket{\rho_{1,j}}\propto \int_{\R^d} e^{-\frac{1}{4}\|x\|^2/\sigma_1^2} \ket{x} \d{x}$\\
  }
  $\widetilde{Z}\gets (2\pi\sigma_1^2)^{d/2}$\\
  \For{$i\gets 1,\dots, M$}{
    $\widetilde{g}_i\gets $ Non-destructive mean estimation for $g_i$ using $\{\ket{\rho_{i, 0}},\dots, \ket{\rho_{i,K}}\}$ \\
    $\widetilde{Z}\gets \widetilde{Z} \widetilde{g}_i$\\
    \For{$j\gets 1,\dots,K$}{
      $\ket{\rho_{i+1,j}}\gets \textsc{QuantumMALA}({\cal O}_{f_{i+1}}, {\cal O}_{\nabla, f_{i+1}}, \ket{\rho_{i,j}})$
    }
  }
  \Return $\widetilde{Z}$
  \caption{\textsc{QuantumMALAforEstimatingNormalizingConstant} (Informal)}
  \label{algo:informal_MALA_est}
\end{algorithm}

\paragraph{Quantum speedup for MLMC}
Now we consider using multilevel Monte Carlo (MLMC) as the annealing process and show how to achieve quantum speedup. MLMC was originally developed by \citet{Hei01} for parametric integration; then \citet{Gil08} applied MLMC to simulate stochastic differential equations (SDEs). The idea of MLMC is natural: we choose a different number of samples at each stage based on the cost and variance of that stage.

To estimate normalizing constants, a variant of MLMC was proposed in~\citet{ge2020estimating}. Unlike the standard MLMC for bounding the mean-squared error, they upper bound the bias and the variance separately, and the analysis is technically difficult. The first quantum algorithm based on MLMC was subsequently developed by~\citet{ALL20} based on the quantum mean estimation algorithm. Roughly speaking, the quantum algorithm can quadratically reduce the $\epsilon$-dependence of the sample complexity compared with classical MLMC.

In this work, we apply the quantum accelerated MLMC (QA-MLMC) scheme \citep{ALL20} to simulate underdamped Langevin dynamics as the SDE. One challenge in using QA-MLMC is that $g_i$ in our setting is not Lipschitz. Fortunately, as suggested by \citet{ge2020estimating}, this issue can be resolved by truncating large $x$ and replacing $g_i$ by $h_i := \min \bigl\{g_i , \exp\bigl(\frac{(r^+_i)^2}{\sigma_i^2(1+\alpha^{-1})}\bigr)\bigr\}$, with the choice
\begin{align}
\alpha = \widetilde O\biggl(\frac{1}{\sqrt{d}\log(1/\epsilon)}\biggr) \qquad
r_i^+ = \E_{\rho_{i+1}}\|x\| + \Theta(\sigma_i\sqrt{(1+\alpha)\log(1/\epsilon)})
\end{align}
to ensure $\frac{h_i}{\E_{\rho_i}g_i}$ is $O(\sigma_i^{-1})$ Lipschitz. Furthermore, $\bigl|\E_{\rho_i}(h_i-g_i)\bigr|<\epsilon$ by Lemmas C.7 and C.8 in \citet{ge2020estimating}. For simplicity, we regard $g_i$ as a Lipschitz continuous function in our main results. We present QA-MLMC in \algo{informal_QA_MLMC}, where the sampling algorithm \textsc{A} can be chosen to be quantum inexact ULD/ULD-RMM or quantum MALA. %

\begin{algorithm}[htbp]
  \KwInput{Evaluation oracle ${\cal O}_f$, function $g$, error $\epsilon$, a quantum sampler \textsc{A}($x_0, f, \eta$) for $\rho$}
  \KwOutput{An estimate of $\widetilde{R} = \E_{\rho}h$}
  $K\gets \widetilde{O}(\epsilon^{-1})$\\
  Compute the initial point $x_0$ and the step size $\eta_0$\\
  Compute the number of samples $N_1,\dots, N_K$\\
  \For{$j\gets 1,\dots,K$}{
    Let $\eta_j=\eta/2^{j-1}$\\
    \For{$i\gets 1,\dots,N_j$}{
      Sample $X_i^{\eta_j}$ by \textsc{A}($f, x_0, \eta_j$), and sample $X_i^{\eta_j/2}$ by \textsc{A}($f, x_0, \eta_j/2$)\\
      $\widetilde{G}_i^- \gets \textsc{QMeanEst}(\{g(X_i^{\eta_j})\}_{i\in [N_j]})$, and $\widetilde{G}_i^+ \gets \textsc{QMeanEst}(\{g(X_i^{\eta_j/2})\}_{i\in [N_j]})$
    }
  }
  \Return $\widetilde{R} = \widetilde{G}_0 + \sum_{j=0}^K(\widetilde{G}_i^- - \widetilde{G}_i^+)$
  \caption{\textsc{QA-MLMC} (Informal)}
  \label{algo:informal_QA_MLMC}
\end{algorithm}

This QA-MLMC framework reduces the $\epsilon$-dependence of the sampling complexity for estimating normalizing constants from $\epsilon^{-2}$ to $\epsilon^{-1}$ in both the ULD and ULD-RMM cases, as compared with the state-of-the-art classical results \cite{ge2020estimating}.

Using the quantum inexact ULD and ULD-RMM algorithms (\thm{informal_inexact}) to generate samples, we obtain our second quantum algorithm for estimating normalizing constants (see \append{estimating_nc_full} for proofs).
\begin{theorem}[Informal version of \thm{quantum-annealing-multilevel-ULD} and \thm{quantum-annealing-multilevel-ULD-RMM}]\label{thm:informal_ULD_est}
  Let $Z$ be the normalizing constant in \prb{log-concave}. There exist quantum algorithms for estimating $Z$ with relative error at most $\epsilon$ using
  \begin{itemize}[leftmargin=*,nosep]
    \item quantum inexact ULD with $\widetilde O(d^{3/2}\kappa^{2}\epsilon^{-1})$
    queries to the evaluation oracle, and
    \item quantum inexact ULD-RMM with $\widetilde O((d^{7/6}\kappa^{7/6}+d^{4/3}\kappa)\epsilon^{-1})$
    queries to the evaluation oracle.
  \end{itemize}
\end{theorem}

\section{Quantum Lower Bound}
Finally, we lower bound the quantum query complexity of normalizing constant estimation.
\begin{theorem}\label{thm:lower-bound}
For any fixed positive integer $k$, given query access \eqn{oracle-evaluation} to a function $f\colon\R^{k}\to\R$ that is 1.5-smooth and 0.5-strongly convex, the quantum query complexity of estimating the partition function $Z=\int_{\R^{k}}e^{-f(x)}\,\d x$ within multiplicative error $\epsilon$ with probability at least $2/3$ is $\Omega(\epsilon^{-\frac{1}{1+4/k}})$.
\end{theorem}

The proof of our quantum lower bound is inspired by the construction in Section 5 of \citet{ge2020estimating}. They consider a log-concave function whose value is negligible outside a hypercube centered at $0$. The interior of the hypercube is decomposed into cells of two types. The function takes different values on each type, and the normalizing constant estimation problem reduces to determining the number of cells of each type. Quantumly, we follow the same construction and reduce the cell counting problem to the \emph{Hamming weight problem}: given an $n$-bit Boolean string and two integers $\ell<\ell'$, decide whether the Hamming weight (i.e., the number of ones) of this string is $\ell_1$ or $\ell_2$. This problem has a known quantum query lower bound~\citep{nayak1999quantum}, which implies the quantum hardness of estimating the normalizing constant. The full proof of \thm{lower-bound} appears in \append{quantum-lower}.

\section*{Acknowledgements}
AMC acknowledges support from the Army Research Office (grant W911NF-20-1-0015); the Department of Energy, Office of Science, Office of Advanced Scientific Computing Research, Accelerated Research in Quantum Computing program; and the National Science Foundation (grant CCF-1813814).
TL was supported by a startup fund from Peking University, and the Advanced Institute of Information Technology, Peking University.
JPL was supported by the National Science Foundation (grant CCF-1813814), an NSF Quantum Information Science and Engineering Network (QISE-NET) triplet award (DMR-1747426), a Simons Foundation award (No. 825053), and the Simons Quantum Postdoctoral Fellowship.
RZ was supported by the University Graduate Continuing Fellowship from UT Austin.

\small
\providecommand{\bysame}{\leavevmode\hbox to3em{\hrulefill}\thinspace}

\newpage
\appendix
\onecolumn

\section{Related Work}\label{append:related-work}
\subsection{Classical MCMC methods}
Our quantum algorithms are inspired by a major class of classical MCMC algorithms based on \emph{Langevin dynamics}. There has been extensive work on non-asymptotic error bounds for the mixing times of Langevin-type algorithms for sampling \citep{DCWY18,sl19,chen2019fast,lst20,WSC21}. One commonly used type of algorithm is based on the mixing time of Langevin dynamics, including the underdamped Langevin diffusion process described by the stochastic differential equations %
\begin{align}\label{eqn:underdamped-Langevin}
\d{v_t} &= - \gamma v_t\,\d t - u\nabla f(x_t)\,\d t + \sqrt{2\gamma u}\,\d W_t \\
\d{x_t} &= v_t\,\d t
\end{align}
with parameters $\gamma, u$, where $W_t \sim \mathcal{N}(0, t)$ is a standard Wiener process.
The coefficients of \eqn{underdamped-Langevin} are Lipschitz continuous since $f$ is $L$-smooth; and the overdamped Langevin diffusion process
\begin{align}\label{eqn:overdamped-Langevin}
\d{x_t} =  -u\nabla f(x_t)\,\d t + \sqrt{2u}\,\d W_t
\end{align}
is obtained by taking $\gamma \to \infty$ and $t \to t/\gamma$.

It can be shown that taking $\gamma=2$ and $u=1/L$, the stationary distribution of the underdamped Langevin diffusion \eqn{underdamped-Langevin} is proportional to $e^{-(f(x)+L\|v\|^2/2)}$, and the marginal distribution of $x$ is proportional to $e^{-f(x)}$. When $\gamma \to \infty$, the stationary distribution of the overdamped version \eqn{overdamped-Langevin} is proportional to $e^{-f(x)}$. The numerical discretization of \eqn{overdamped-Langevin} is used in unadjusted Langevin algorithms, while sampling algorithms based on the discretization of \eqn{underdamped-Langevin} can have a better dependence on $d$ and $\epsilon$.

We now introduce a few common classical sampling algorithms: %
the underdamped Langevin diffusion (ULD) method; the randomized midpoint method for underdamped Langevin diffusion (ULD-RMM), with the best known dependence on $d$; and the Metropolis adjusted Langevin algorithm (MALA), with the best known dependence on $\kappa$ and $\epsilon$. To simulate the random process in discrete time, ULD can be viewed as the first-order forward Euler discretization of the continuous process \eqn{underdamped-Langevin}. In particular, ULD takes $\widetilde O\bigl({\kappa^2\sqrt{d}}/{\epsilon}\bigr)$ steps to approximate the stationary distribution $e^{-f(x)}$ within $\epsilon$ in the Wasserstein 2-norm~\citep{CCBJ18}, where $\kappa$ is the condition number of $f$, and $d$ is the dimension. ULD-RMM approximates the integral of the random process \eqn{underdamped-Langevin} by randomly choosing the midpoint in the integral, which reduces the bias in the accumulation of the integration. As a more accurate approximation, ULD-RMM converges in the Wasserstein 2-distance $\epsilon$ with $\widetilde O\Bigl(\frac{\kappa^{7/6}d^{1/6}}{\epsilon^{1/3}}+\frac{\kappa d^{1/3}}{\epsilon^{2/3}}\Bigr)$ steps~\citep{sl19}, a polynomial reduction in $\kappa, d, \epsilon$ over ULD. As an alternative approach, MALA also constructs the Euler discretization of \eqn{underdamped-Langevin}, and then applies the Metropolis-Hastings acceptance/rejection step to ensure convergence to the correct stationary distribution. It was first shown that MALA converges in the total variation distance $\epsilon$ with $\widetilde O\bigl(\kappa d\max\{1,\kappa/d\}\log(\kappa d/\epsilon)\bigr)$ steps for Gaussian initial distributions~\citep{DCWY18,chen2019fast}. Later, this result was improved to $\widetilde O\bigl(\kappa d\log(\kappa d/\epsilon)\bigr)$ based on an improved non-asymptotic analysis of the mixing time~\citep{lst20}. For warm-start distributions, the complexity of MALA can be further reduced to $\widetilde O\bigl(\kappa d^{1/2}\log(\kappa d/\epsilon)\bigr)$~\citep{WSC21}. Compared to ULD and ULD-RMM, this exponentially improves the dependence on $\epsilon$, and polynomially improves the dependence on $\kappa$, while it suffers from a worse dependence on $d$. We introduce the algorithms and complexities of ULD and ULD-RMM in \append{classical_ULD}, and introduce these results of MALA in \append{classical_MALA}.

For the task of estimating the normalizing constant \prb{log-concave}, the state-of-the-art classical results are given by \citet{ge2020estimating}. That work applies the classical sampling algorithms described above with an annealing strategy. The normalizing constant is estimated by a sequence of telescoping sums, each of which can be approximated by a Monte Carlo method that samples from a log-concave distribution. We introduce this annealing procedure in \append{annealing}. Reference~\citet{ge2020estimating} employed the mixing time of MALA for Gaussian initial distributions developed by~\citep{DCWY18,chen2019fast} with the annealing procedure, achieving the overall complexity $\widetilde O\Bigl(\frac{\kappa d^2}{\epsilon^2}\max\{1,\kappa/d\}\Bigr)$ for estimating the normalizing constant. They also combined ULD and ULD-RMM with the annealing and the multilevel Monte Carlo (MLMC) method to achieve complexities of $\widetilde O\Bigl(\frac{\kappa^2d^{3/2}}{\epsilon^2}\Bigr)$ and $\widetilde O\Bigl(\frac{\kappa^{7/6}d^{7/6}}{\epsilon^2}+\frac{\kappa d^{4/3}}{\epsilon^2}\Bigr)$, respectively. Here MLMC, introduced in \append{quantum-multilevel}, is utilized to resolve the worse dependence on $\epsilon$ in ULD and ULD-RMM, resulting in the same $\widetilde O(1/\epsilon^2)$ scaling of the error compared to the annealing with MALA. Annealing with MLMC and ULD/ULD-RMM also has a better dependence on $d$ over annealing with MALA, while they suffer from a worse dependence on $\kappa$.

\subsection{Quantum computing}\label{append:related-work-quantum}
Previous literature developed alternative approaches to generating quantum states corresponding to classical probability distributions on a quantum computer, sometimes referred to as quantum sampling (or qsampling) from a distribution. References \citet{zalka1998simulating}, \citet{GR02}, and \citet{KM01} propose direct state generation approaches using controlled rotations. However, this approach is limited to the regime in which the distribution is efficiently integrable. As an alternative, \citet{aharonov2003adiabatic} develops an adiabatic approach to qsampling. They apply adiabatic evolution techniques to qsample the stationary distributions of a sequence of slowly varying Markov chains, a technique referred to as quantum simulated annealing (QSA) in subsequent literature \citep{somma2007quantum,somma2008quantum,wocjan2008speedup,yung2012quantum,harrow2019adaptive}. The time complexity of Aharanov and Ta-Shma's approach is $O(1/\delta)$ as a function of the spectral gap $\delta$, comparable to the running time of analogous classical sampling methods. Reference \citet{wocjan2008speedup} adopted Szegedy's quantum walks \citep{szegedy2004quantum} and amplitude amplification \citep{brassard2002amplitude} to improve the time complexity of this qsampling procedure to $O(1/\sqrt{\delta})$, achieving a quadratic speedup in the spectral gap. As a generalization, \citet{temme2011quantum} proposes a quantum Metropolis sampling method that extends qsampling to quantum Hamiltonians, with time complexity $O(1/\delta)$. Reference \citet{yung2012quantum} combines quantum Metropolis sampling with QSA to achieve time complexity $O(1/\sqrt{\delta})$. Another alternative approach is quantum rejection sampling \citep{orr13,low2014quantum,wiebe2015can}, which provides a method for transforming an initial superposition of desired and undesired states into the desired state using amplitude amplification. Reference \citet{wiebe2015can} employs semi-classical Bayesian updating to achieve time complexity $O(1/\sqrt{\epsilon})$ as a function of the approximation error $\epsilon$. The quantum rejection sampling approach is generally less efficient than the QSA approach, as the latter can achieve $O(\log 1/\epsilon)$ by choosing proper slowly varying MCs that mix rapidly.

Previous quantum computing literature on partition function estimation mainly focused on discrete systems with
\begin{align}\label{eqn:partition-discrete}
Z(\beta) = \sum_{x\in\Omega} e^{-\beta H(x)},
\end{align}
where $\beta$ is an inverse temperature and $H$ is a Hamiltonian function of $x$ over a finite state space $\Omega$. The space $\Omega$ is usually assumed to be a simple discrete set, such as $\{0,1\}^n$, and $H$ is assumed to be a sum of local terms. For instance, \citet{montanaro2015quantum} considers $H$ taking integer values $\{0,1, \ldots, n\}$, and \citet{harrow2019adaptive} assumes $0\le H(x)\le n$ for all $x$.

To estimate $Z = Z(\infty)$ in \eqn{partition-discrete}, \citet{montanaro2015quantum} considers a classical Chebyshev cooling schedule $0 = \beta_0 < \beta_1 < \ldots \beta_l = \infty$ for $Z$ with the length $l = O(\sqrt{\log |\Omega|}\log\log |\Omega|)$ \citep{SVV2009}. %
Reference \citet{montanaro2015quantum} applies fast qsampling algorithms to estimate $Z$ with $\widetilde O(l^2/\sqrt{\delta}\epsilon) = \widetilde O(\log |\Omega|/\sqrt{\delta}\epsilon)$ quantum walk steps to sample from Gibbs distributions $\pi_i(x) = \frac{1}{Z(\beta_i)}e^{-\beta_i H(x)}$, whereas a corresponding classical algorithm takes $\widetilde O(l^2/\delta\epsilon^2) = \widetilde O(\log |\Omega|/\delta\epsilon^2)$ random walk steps. Here $\epsilon$ denotes the relative error for estimating $Z$, and $\delta$ denotes the spectral gap of the Markov chains with stationary distributions $\pi_i(x)$.
Reference \citet{montanaro2015quantum} also points out that this quantum algorithm relies on classical Markov chain Monte Carlo for computing the Chebyshev cooling schedule, introducing an overhead of $\widetilde O(\log |\Omega|/\delta)$ \citep{SVV2009}. Hence, the overall cost is $\widetilde O(\log |\Omega|/\sqrt{\delta}(\epsilon+\sqrt{\delta}))$, a quadratic reduction with respect to $\epsilon$ over classical algorithms.
Reference \citet{harrow2019adaptive} develops a fully quantized version of the Chebyshev cooling schedule that only requires additional cost $\widetilde O(\log |\Omega|/\sqrt{\delta})$. This results in overall cost $\widetilde O(\log |\Omega|/\sqrt{\delta}\epsilon)$, a quadratic speedup in terms of $\delta$ over \citet{montanaro2015quantum} and classical algorithms.
Reference \citet{AHN21} constructs a shorter Chebyshev cooling schedule by using a paired-product estimator with length $l = O(\sqrt{\log |\Omega|})$, eliminating the $l = O(\log\log |\Omega|)$ factors in the previous schedule \citep{SVV2009}. Reference~\citet{AHN21} develops a fully quantized version of this shorter schedule, almost matching the same overall cost $\widetilde O(\log |\Omega|/\sqrt{\delta}(\epsilon+\sqrt{\delta}))$ of \citet{harrow2019adaptive}.

Estimating the partition function of a discrete system corresponds to a discrete counting problem, with applications such as counting colorings or matchings of a graph and estimating statistics of Ising models, while estimating partition functions of continuous systems is relevant to the volume estimation problem.

\section{Tools from Classical MCMC Algorithms}
\subsection{ULD and ULD-RMM}\label{append:classical_ULD}

We now describe underdamped Langevin diffusion (ULD) and the randomized midpoint method for underdamped Langevin diffusion (ULD-RMM), as introduced in~\citet{ge2020estimating} with Lipschitz continuous constants. We consider the underdamped Langevin diffusion with parameters $\gamma, u$:
\begin{align}
\d{v_t} &= - \gamma v_t\,\d t - \nabla f(x_t)\,\d t + \sqrt{2\gamma u}\,\d W_t, \\
\d{x_t} &= v_t\,\d t,
\end{align}
The discrete dynamics of ULD can be described by
\begin{align}\label{eqn:ULD-discrete}
\d{v_t^h} &= - \gamma v_t^h\,\d t - u\nabla f(x_{\lfloor t/h\rfloor h}^h)\,\d t + \sqrt{2\gamma u}\,\d W_t, \\
\d{x_t^h} &= v_t^h\,\d t.
\end{align}
According to~\citet{ge2020estimating}, taking $\gamma=2$ and $u=1/L$, the explicit discrete-time update of ULD is integrated as
\begin{align}\label{eqn:ULD}
v_{t+h}^h &= e^{-2h}v_t^h + \frac{1}{2L}(1-e^{-2h})\nabla f(x_t^h) + \frac{2}{\sqrt{L}}W_{1,t}^h, \\
x_{t+h}^h &= x_t^h + \frac{1}{2}(1-e^{-2h})v_t^h + \frac{1}{2L}[h-(1-e^{-2h})]\nabla f(x_t^h) + \frac{1}{\sqrt{L}}W_{2,t}^h,
\end{align}
where
\begin{align}\label{eqn:ULD-Brownian}
W_{1,t}^h &= \int_0^he^{2(s-h)}\d B_{t+s}, \\
W_{2,t}^h &= \int_0^h(1-e^{2(s-h)})\d B_{t+s}.
\end{align}
$W_{1,t}^h$ and $W_{2,t}^h$ can be obtained by sampling the $d$-dimensional standard Brownian motion $B_t$.

The ULD algorithm is presented in \algo{ULD}. The convergence of ULD has been established by Theorem 1 of \citet{CCBJ18}, which was restated by Theorem C.3 of \citet{ge2020estimating} as follows.

\begin{algorithm}[htbp]
\KwInput{Function $f$, step size $h$, time $T$, and a sample $x_0$ from a starting distribution $\rho_0$}
\KwOutput{Sequence $x_h^h, x_{2h}^h, \ldots, x_{\lfloor T\rfloor + 1}^h$}
Compute $x_0^h\leftarrow x_0$\\
\For{$t=0, h, \ldots, \lfloor T\rfloor$}{
Draw $W_{1,t}^h = \int_0^he^{2(s-h)}\d B_{t+s}$, $W_{2,t}^h = \int_0^h(1-e^{2(s-h)})\d B_{t+s}$\\
Compute
$v_{t+h}^h = e^{-2h}v_t^h + \frac{1}{2L}(1-e^{-2h})\nabla f(x_t^h) + \frac{2}{\sqrt{L}}W_{1,t}^h$,
$x_{t+h}^h = x_t^h + \frac{1}{2}(1-e^{-2h})v_t^h + \frac{1}{2L}[h-(1-e^{-2h})]\nabla f(x_t^h) + \frac{1}{\sqrt{L}}W_{2,t}^h$
}
\caption{Underdamped Langevin Dynamics (ULD)}
\label{algo:ULD}
\end{algorithm}

\begin{lemma}[{Theorem 1 of \citet{CCBJ18}}]\label{lem:ULD}
Assume the target distribution $\rho$ is strongly log-concave with $L$-smooth and $\mu$-strongly convex negative log-density. Let $\rho_n$ be the distribution of the underdamped Langevin diffusion with the initial point $x_0$ satisfying $\|x_0-x^{\ast}\|\le D$, step size $h\le\frac{\epsilon}{104\kappa}\sqrt{\frac{1}{d/\mu+D^2}}$, and time $T\ge\frac{\kappa}{2}\log\Bigl(\frac{24\sqrt{d/\mu+D^2}}{\epsilon}\Bigr)$. Then ULD achieves
\begin{align}
\E\left(\|\widehat{X}_n-X_T\|^2\right) &\leq \widetilde O\Big(\frac{d^2\kappa^2h^2}{\mu}\Big), \label{eqn:ULD-order} \\
W_2(\rho_n, \rho) &\leq \epsilon, \label{eqn:ULD-error}
\end{align}
using
\begin{align}\label{eqn:ULD-query}
\frac{T}{h}=\widetilde \Theta\Big(\frac{\kappa^2\sqrt{d}}{\epsilon}\Big)
\end{align}
queries to $\nabla f$.
\end{lemma}

According to~\citet{ge2020estimating}, the explicit discrete-time update of ULD-RMM is integrated as
\begin{align}\label{eqn:ULD-RMM}
v_{t+h}^h &= e^{-2h}v_t^h + \frac{h}{L}e^{-2(1-\alpha)h}\nabla f(y_t^h) + \frac{2}{\sqrt{L}}W_{1,t}^h, \\
x_{t+h}^h &= x_t^h + \frac{1}{2}(1-e^{-2h})v_t^h + \frac{h}{2L}(1-e^{-2(1-\alpha)h})\nabla f(y_t^h) + \frac{1}{\sqrt{L}}W_{2,t}^h, \\
y_{t+h}^h &= x_t^h + \frac{1}{2}(1-e^{-2\alpha h})v_t^h + \frac{1}{2L}[\alpha h-(1-e^{-2\alpha h})]\nabla f(x_t^h) + \frac{1}{\sqrt{L}}W_{3,t}^h,
\end{align}
where
\begin{align}\label{eqn:ULD-RMM-Brownian}
W_{1,t}^h &= \int_0^he^{2(s-h)}\d B_{t+s}, \\
W_{2,t}^h &= \int_0^h(1-e^{2(s-h)})\d B_{t+s}, \\
W_{3,t}^h &= \int_0^{\alpha h}(1-e^{2(s-h)})\d B_{t+s}.
\end{align}
$W_{1,t}^h$, $W_{2,t}^h$, and $W_{3,t}^h$ can be obtained by sampling the $d$-dimensional standard Brownian motion $B_t$.

The ULD-RMM algorithm is presented in \algo{ULD-RMM}. The convergence of ULD-RMM has been established by Theorem 3 of \citet{sl19}, which was restated by Theorem C.5 of \citet{ge2020estimating} as follows.

\begin{algorithm}[htbp]
\KwInput{Function $f$, step size $h$, time $T$, and a sample $x_0$ from a starting distribution $\rho_0$}
\KwOutput{Sequence $x_h^h, x_{2h}^h, \ldots, x_{\lfloor T\rfloor + 1}^h$}
Compute $x_0^h\leftarrow x_0$, $y_0^h\leftarrow x_0$\\
\For{$t=0, h, \ldots, \lfloor T\rfloor$}{
Draw $W_{1,t}^h = \int_0^he^{2(s-h)}\d B_{t+s}$, $W_{2,t}^h = \int_0^h(1-e^{2(s-h)})\d B_{t+s}$, $W_{3,t}^h = \int_0^{\alpha h}(1-e^{2(s-h)})\d B_{t+s}$\\
Compute
$v_{t+h}^h = e^{-2h}v_t^h + \frac{h}{L}e^{-2(1-\alpha)h}\nabla f(y_t^h) + \frac{2}{\sqrt{L}}W_{1,t}^h$,
$x_{t+h}^h = x_t^h + \frac{1}{2}(1-e^{-2h})v_t^h + \frac{h}{2L}(1-e^{-2(1-\alpha)h})\nabla f(y_t^h) + \frac{1}{\sqrt{L}}W_{2,t}^h$,
$y_{t+h}^h = x_t^h + \frac{1}{2}(1-e^{-2\alpha h})v_t^h + \frac{1}{2L}[\alpha h-(1-e^{-2\alpha h})]\nabla f(x_t^h) + \frac{1}{\sqrt{L}}W_{3,t}^h$
}
\caption{Underdamped Langevin Dynamics with Randomized Midpoint Method (ULD-RMM)}
\label{algo:ULD-RMM}
\end{algorithm}

\begin{lemma}[{Theorem 3 of \citet{sl19}}]\label{lem:ULD-RMM}
Assume the target distribution $\rho$ is strongly log-concave with $L$-smooth and $\mu$-strongly convex negative log-density. Let $\rho_n$ be the distribution of the randomized midpoint method for underdamped Langevin diffusion with the initial point $x_0$, step size $h\le\min\Bigl\{\frac{\epsilon^{1/3}\mu^{1/6}}{\kappa^{1/6}d^{1/6}\log^{1/6}\bigl(\frac{\sqrt{d/\mu}}{\epsilon}\bigr)},\frac{\epsilon^{2/3}\mu^{1/3}}{d^{1/3}\log^{1/3}\bigl(\frac{\sqrt{d/\mu}}{\epsilon}\bigr)}\Bigr\}$, and time $T\ge2\kappa\log\Bigl(\frac{20d/\mu}{\epsilon^2}\Bigr)$. Then ULD-RMM achieves
\begin{align}
\E\left(\|\widehat{X}_n-X_T\|^2\right) &\leq \widetilde O\Big(\frac{d\kappa h^6}{\mu}+\frac{dh^3}{\mu}\Big), \label{eqn:ULD-RMM-order} \\
W_2(\rho_n, \rho) &\leq \epsilon, \label{eqn:ULD-RMM-error}
\end{align}
using
\begin{align}\label{eqn:ULD-RMM-query}
\frac{2T}{h}=\widetilde \Theta\Big(\frac{\kappa^{7/6}d^{1/6}}{\epsilon^{1/3}}+\frac{\kappa d^{1/3}}{\epsilon^{2/3}}\Big)
\end{align}
queries to $\nabla f$.
\end{lemma}

\subsection{Annealing for Estimating the Normalizing Constant}\label{append:annealing}

Having described the sampling procedure for a log-concave function, we now move to the problem of estimating the normalizing constant
\begin{align}
Z=\int_{x\in \R^d}e^{-f(x)}\d x.
\end{align}

We consider a sequence of auxiliary distributions, given by
\begin{align}\label{eqn:distribution-f}
f_i(x) = \frac{1}{2}\frac{\|x\|^2}{\sigma_i^2} + f(x)
\end{align}
for $i\in\range{M}$, where $\sigma_1\le\sigma_2\le\cdots\le\sigma_M$. We define $\sigma_{M+1}=\infty$ and $f_{M+1}=f$ for convenience. We consider the sequence of distributions
\begin{align}\label{eqn:distribution-rho}
\rho_i(\d x) = Z_i^{-1}e^{-f_i(x)}\d x,
\end{align}
where $Z_i$ is the normalizing constant
\begin{align}\label{eqn:distribution-Z}
Z_i = \int_{x\in \R^d}e^{-f_i(x)}\d x.
\end{align}
Then $Z$ is estimated by the telescoping product
\begin{align}\label{eqn:telescoping-Z}
Z = Z_{M+1} = Z_1\prod_{i=1}^{M}\frac{Z_{i+1}}{Z_{i}}.
\end{align}
In \eqn{telescoping-Z}, we first approximate $Z_1$ by the normalizing constant of the Gaussian distribution with variance $\sigma_1^2$, which is bounded by the following lemma.

\begin{lemma}[{Lemma 3.1 of \citet{ge2020estimating}}]\label{lem:telescoping-1}
Letting $\sigma_1^2=\frac{\epsilon}{2dL}$, we have
\begin{align}\label{eqn:telescoping-Z-1}
\Bigl(1-\frac{\epsilon}{2}\Bigr)\int_{x\in \R^d}e^{-\frac{1}{2}\frac{\|x\|^2}{\sigma_1^2}}\d x \le Z_1 \le \int_{x\in \R^d}e^{-\frac{1}{2}\frac{\|x\|^2}{\sigma_1^2}}\d x.
\end{align}
\end{lemma}

We then approximate $\frac{Z_{i+1}}{Z_{i}}$ by sampling the distribution $\rho_i$, with
\begin{align}\label{eqn:telescoping-sampling}
\frac{Z_{i+1}}{Z_{i}} = \E_{\rho_i}(g_i),
\end{align}
where
\begin{align}\label{eqn:distribution-g}
g_i = \exp\Bigl(\frac{1}{2}\bigl(\frac{1}{\sigma_{i}^2}-\frac{1}{\sigma_{i+1}^2}\bigr)\|x\|^2\Bigr).
\end{align}
If $X_i^{(1)}, X_i^{(2)}, \ldots, X_i^{(K)}$ are i.i.d.~samples generated according to the distribution $\rho_i$, then
\begin{align}\label{eqn:telescoping-estimate}
\frac{Z_{i+1}}{Z_{i}} \approx \frac{1}{K}\sum_{k=1}^Kg_i(X_i^{(k)}).
\end{align}

For the sequence of $\sigma_i^2$ with the annealing schedule $\frac{\sigma_{i+1}^2}{\sigma_{i}^2} = 1+\alpha$, we aim to bound the relative variance of $\frac{Z_{i+1}}{Z_{i}}$. First, for $\frac{Z_{M+1}}{Z_{M}}$, we have the following lemma.

\begin{lemma}[{Lemma 3.2 of \citet{ge2020estimating}}]\label{lem:telescoping-M}
For $\sigma_M^2\ge\frac{2}{\mu}$, we have
\begin{align}\label{eqn:relative-variance-M}
\frac{\E_{\rho_M}(g_M^2)}{\E_{\rho_M}(g_M)^2} \le \exp(\frac{4d}{\mu\sigma_M^4}).
\end{align}
\end{lemma}
When $\sigma_M^2\ge\frac{2\sqrt{d}}{\mu}$, and assuming $\mu<1$, we have $\frac{\E_{\rho_M}(g_M^2)}{\E_{\rho_M}(g_M)^2} \le e$.

Second, for $\frac{Z_{i+1}}{Z_{i}}$ with $i\in\range{M-1}$, the relative variance of can be bounded by the following lemma.

\begin{lemma}[{[Lemma 3.3 of \citet{ge2020estimating}}]\label{lem:telescoping}
Let $\rho$ be a log-concave distribution. For $\alpha\le\frac{1}{2}$, we have
\begin{align}\label{eqn:relative-variance}
\frac{\E_{\rho_i}(g_i^2)}{\E_{\rho_i}(g_i)^2} = \frac{\E_\rho\left[\exp(-\frac{1+\alpha}{2}\frac{\|x\|^2}{\sigma^2})\right]\cdot \E_\rho\left[\exp(-\frac{1-\alpha}{2}\frac{\|x\|^2}{\sigma^2})\right]}{\E_\rho\left[\exp(-\frac{1}{2}\frac{\|x\|^2}{\sigma^2})\right]^2} \le \exp(4\alpha^2d).
\end{align}
\end{lemma}
Therefore, if we choose the annealing schedule
\begin{align}\label{eqn:telescoping-sigma}
\frac{\sigma_{i+1}^2}{\sigma_{i}^2} = 1+\frac{1}{\sqrt{d}},
\end{align}
then $\frac{\E_{\rho_i}(g_i^2)}{\E_{\rho_i}(g_i)^2} \le e^4$.

The estimate of the normalizing constant \prb{log-concave} relies on the above annealing framework and the sampling algorithms for the log-concave distribution $\rho_i$ including ULD, ULD-RMM, and MALA. In the following sections, we discuss the quantum speedup for \prb{log-concave} using MALA and annealing, and using multilevel ULD/ULD-RMM and annealing.

\subsection{Annealing Markov chains are slowly varying}
The goal of this subsection is to prove the following lemma.
\begin{lemma}[Slowly varying MCs]\label{lem:slowly_varying}
Let $f_0(x) = \frac{\|x\|^2}{2\sigma_1^2}$ and let $\d \pi_0=(2\pi \sigma_1^2)^{d/2}\cdot e^{-f_0(x)}\d x$ be the Gaussian distribution.
For $i\in \{1,\dots,M\}$, let $f_i(x)=f(x) + \frac{\|x\|^2}{2\sigma_i^2}$ and let $\d \pi_i=Z_i^{-1}e^{-f_i(x)}\d x$ be its stationary distribution. Let $f_{M+1}(x)=f(x)$ and let $\d \pi_{M+1}$ be the target log-concave distribution. Define the qsample state
\begin{align}
  \ket{\pi_i}=\int_\Omega \d x\sqrt{\pi_i(x)}\ket{x} \quad \forall \, 0\leq i\leq M+1.
\end{align}
Then, for $0\leq i\leq M$, we have
\begin{align}
  |\langle \pi_i|\pi_{i+1}\rangle|\geq \Omega(1).
\end{align}
\end{lemma}

\begin{proof}
First, we consider the case when $i=0$. Note that $|\langle \pi_0|\pi_1\rangle|$ can be written as
\begin{align}
  |\langle \pi_0|\pi_1\rangle| = &~ \int_\Omega \d x\cdot (2\pi\sigma_1^2)^{-d/4}e^{-\frac{1}{2}f_0(x)}\cdot Z_1^{-1/2}e^{-\frac{1}{2}f_1(x)}\\
= &~ \frac{\int_\Omega {e^{-\frac{1}{2}f(x)-\frac{\|x\|^2}{2\sigma_1^2}}\d x}}{(2\pi\sigma_1^2)^{d/4}\cdot \sqrt{Z_1}}.
\end{align}
Since $0\leq f(x)\leq \frac{1}{2}L\|x\|^2$, the numerator can be lower bounded by
\begin{align}
  \int_\Omega e^{-\frac{1}{2}f(x)-\frac{\|x\|^2}{2\sigma_1^2}}\d x \geq &~ \int_\Omega e^{-\frac{1}{2}(L/2+\sigma_1^{-2})\|x\|^2} \d x=\left(2\pi(L/2+\sigma_1^{-2})^{-1}\right)^{d/2}
\end{align}
and the denominator can be upper bounded by
\begin{align}
  (2\pi\sigma_1^2)^{d/4}\cdot \sqrt{\int_\Omega e^{-f(x)-\frac{1}{2}\|x\|^2/\sigma_1^{2}} \d x} \leq &~  (2\pi\sigma_1^2)^{d/4}\cdot \sqrt{\int_\Omega e^{-\frac{1}{2}\|x\|^2/\sigma_1^{2}} \d x}= (2\pi\sigma_1^2)^{d/2}.
\end{align}
Therefore
\begin{align}
  |\langle \pi_0|\pi_1\rangle| \geq \frac{\left(2\pi(L/2+\sigma_1^{-2})^{-1}\right)^{d/2}}{(2\pi\sigma_1^2)^{d/2}}= (1+\sigma_1^2L/2)^{-d/2}\geq e^{-\sigma_1^2dL/4}.
\end{align}
By our choice of $\sigma_1^2=\frac{\epsilon}{2dL}$, we have $|\langle \pi_0|\pi_1\rangle|\geq e^{-\epsilon/8}=\Omega(1)$.

Now consider the case where $1\leq i\leq M-1$.
  The inner product between $\ket{\pi_i}$ and $\ket{\pi_{i+1}}$ can be written as
  \begin{align}
    |\langle \pi_i|\pi_{i+1}\rangle| = &~ \int_{\Omega} \d x \cdot Z_i^{-1/2} e^{-\frac{1}{2}f_i(x)}\cdot Z_{i+1}^{-1/2} e^{-\frac{1}{2}f_{i+1}(x)}\\
    = &~ \frac{\int_\Omega e^{-f(x)-\frac{\|x\|^2}{4}(\sigma_i^{-2}+\sigma_{i+1}^{-2})}\d x}{\sqrt{Z_iZ_{i+1}}}.
  \end{align}
  Let $\sigma^2=\sigma_{i+1}^2$ and $\sigma^2/(1+\alpha) = \sigma_i^2$. Also, let $\rho$ be the log-concave distribution $\rho(\d x) = Z^{-1}e^{-f(x)}\d x$. Then we have
  \begin{align}
    \int_\Omega e^{-f(x)-\frac{\|x\|^2}{4}(\sigma_i^{-2}+\sigma_{i+1}^{-2})}\d x = Z\cdot \E_\rho\left[e^{-\frac{1+\alpha/2}{2\sigma^2}\|x\|^2}\right].
  \end{align}
  Similarly,
  \begin{align}
    Z_i = Z\cdot \E_\rho\left[e^{-\frac{1+\alpha}{2\sigma^2}\|x\|^2}\right] \quad\text{and}\quad
    Z_{i+1} = Z\cdot \E_\rho\left[e^{-\frac{1}{2\sigma^2}\|x\|^2}\right].
  \end{align}
  Hence,
  \begin{align}
    |\langle \pi_i|\pi_{i+1}\rangle| = &~ \frac{\E_\rho\left[e^{-\frac{1+\alpha/2}{2\sigma^2}\|x\|^2}\right]}{\E_\rho\left[e^{-\frac{1+\alpha}{2\sigma^2}\|x\|^2}\right]^{1/2}\cdot \E_\rho\left[e^{-\frac{1}{2\sigma^2}\|x\|^2}\right]^{1/2}}.
  \end{align}
  Let $\alpha':=\frac{\alpha}{\alpha + 2}$ and $\sigma'^2:=\frac{\sigma^2}{1+\alpha/2}$. Then
  \begin{align}
    |\langle \pi_i|\pi_{i+1}\rangle| = &~ \frac{\E_\rho\left[e^{-\frac{1}{2\sigma'^2}\|x\|^2}\right]}{\E_\rho\left[e^{-\frac{1+\alpha'}{2\sigma'^2}\|x\|^2}\right]^{1/2}\cdot \E_\rho\left[e^{-\frac{1-\alpha'}{2\sigma'^2}\|x\|^2}\right]^{1/2}}\\
    \geq &~ e^{-2\alpha'^2 d},
  \end{align}
  where the last step follows from \lem{telescoping}. Since we choose $\alpha = d^{-1/2}$, we have $\alpha' = \frac{1}{1+2\sqrt{d}}=O(d^{-1/2})$, which implies that $e^{-2\alpha^2 d} = \Omega(1)$.

  Finally, we consider the case where $i=M$. The inner product can be written as
  \begin{align}
    |\langle \pi_M|\pi_{M+1}\rangle| = &~ \frac{\int_\Omega\d x \cdot e^{-\frac{1}{2}f(x)-\frac{1}{4}\|x\|^2/\sigma_M^{2}}\cdot e^{-\frac{1}{2}f(x)}}{\sqrt{Z_M}\sqrt{Z}}\\
    = &~ \frac{\int_\Omega\d x \cdot e^{-f(x)-\frac{1}{4}\|x\|^2/\sigma_M^{2}}}{\sqrt{Z_M}\sqrt{Z}}.
  \end{align}
  Let $\rho'$ be a log-concave distribution with density proportional to $e^{-f(x)-\frac{1}{4}\|x\|^2/\sigma_M^{2}}$. Then
  \begin{align}
    \frac{\int_\Omega\d x \cdot e^{-f(x)-\frac{1}{4}\|x\|^2/\sigma_M^{2}}}{\sqrt{Z_M}\sqrt{Z}} = &~ \E_{\rho'}\left[e^{-\frac{1}{4}\|x\|^2/\sigma_M^2}\right]^{-1/2}\cdot \E_{\rho'}\left[e^{\frac{1}{4}\|x\|^2/\sigma_M^2}\right]^{-1/2}\\
    \geq &~ e^{-\frac{d}{2\mu\sigma_M^{4}}},
  \end{align}
  where the second step follows from the proof of \lem{telescoping-M} in \citet{ge2020estimating}.
  Since we take $\sigma_M^2=\Theta(\frac{\sqrt{d}}{\mu})$, we find that $|\langle \pi_M|\pi_{M+1}\rangle| \geq e^{-\Theta(1)}=\Omega(1)$.

  Combining the three cases, the proof is complete.
\end{proof}

\section{Quantum Algorithm for Log-Concave Sampling: Details}\label{append:log-concave_sampling}

In this section, we provide several quantum algorithms for sampling log-concave distributions. In \append{q_ULD_sampling}, we show that the classical underdamped Langevin diffusion (ULD) and the randomized midpoint method for underdamped Langevin diffusion (ULD-RMM) can be improved by replacing the first-order oracle by the zeroth-order quantum oracle, while achieving the same efficiency and accuracy guarantees. In \append{q_MALA_sampling}, we show that the  Metropolis adjusted Langevin algorithm (MALA) can be quantumly sped up in terms of query complexity, for both Gaussian initial distributions and warm-start distributions.

\subsection{Quantum Inexact ULD and ULD-RMM}\label{append:q_ULD_sampling}

In the quantum setting, we can estimate $\nabla f(x)$ by using Jordan's algorithm with queries to the quantum zeroth-order evaluation oracle \eqn{oracle-evaluation}. The following lemma provides an $\ell_1$-error guarantee.

\begin{lemma}[Lemma 2.3 in \cite{chakrabarti2018quantum}]\label{lem:SmoothQuantumGradient}
  Let $f$ be a convex, $L_0$-Lipschitz continuous function that is specified by an evaluation oracle with error at most $\epsilon$. Suppose $f$ is $L$-smooth in $B_{\infty}(x, 2\sqrt{\epsilon/L})$. Let
  \begin{align}
  \widetilde g = \textsf{SmoothQuantumGradient}(f, \epsilon, L_0, L, x).
  \end{align}
  Then for any $i\in\range{d}$, we have $|\widetilde g_i| \le L_0$ and $\E|\widetilde g_i-\nabla f(x)_i| \le 3000\sqrt{d\epsilon L}$; hence
  \begin{align}
  \E\|\widetilde g-\nabla f(x)\|_1 \le 3000d^{1.5}\sqrt{\epsilon L}.
  \end{align}
  If $L_0$, $1/L$, and $1/\epsilon$ are $\poly(d)$, the \textsf{SmoothQuantumGradient} algorithm uses $O(1)$ queries to the quantum evaluation oracle and $\widetilde O(d)$ gates.
  \end{lemma}

We then introduce inexact ULD and ULD-RMM by using a stochastic zeroth-order oracle as follows.

\begin{lemma}[{Theorem 2.2 of \citet{roy2019stochastic}}]\label{lem:IULD}
Let $\rho_n$ be the distribution of the underdamped Langevin diffusion with the initial point $x_0$ satisfying $\|x_0-x^{\ast}\|\le D$, step size $h\le\frac{\epsilon}{104\kappa}\sqrt{\frac{1}{d/\mu+D^2}}$, and time $T\ge\frac{\kappa}{2}\log\Bigl(\frac{24\sqrt{d/\mu+D^2}}{\epsilon}\Bigr)$. Assume there is a stochastic zeroth-order oracle that provides an unbiased evaluation of $\nabla f(x)$ with bounded variance $\E\|\widetilde g-\nabla f(x)\|^2 \le \sigma^2$. Then inexact ULD achieves $W_2(\rho_n, \rho) \leq \epsilon$
using
\begin{align}\label{eqn:IULD-iteration}
\frac{T}{h}=\widetilde \Theta\Big(\frac{\kappa^2\sqrt{d}}{\epsilon}\Big)
\end{align}
iterations and
\begin{align}\label{eqn:IULD-query}
b=\frac{d^{1.5}\max\{1,\sigma^2\}}{\epsilon}
\end{align}
queries to the zeroth-order oracle per iteration. The total number of calls is $\frac{bT}{h}$.
\end{lemma}

\begin{lemma}[{Theorem 2.3 of \citet{roy2019stochastic}}]\label{lem:IULD-RMM}
Let $\rho_n$ be the distribution of the randomized midpoint method for underdamped Langevin diffusion with the initial point $x_0$, step size $h\le\min\Bigl\{\frac{\epsilon^{1/3}\mu^{1/6}}{\kappa^{1/6}d^{1/6}\log^{1/6}\bigl(\frac{\sqrt{d/\mu}}{\epsilon}\bigr)},\frac{\epsilon^{2/3}\mu^{1/3}}{d^{1/3}\log^{1/3}\bigl(\frac{\sqrt{d/\mu}}{\epsilon}\bigr)}\Bigr\}$, and time $T\ge2\kappa\log\Bigl(\frac{20d/\mu}{\epsilon^2}\Bigr)$. Assume there is a stochastic zeroth-order oracle that provides an unbiased evaluation of $\nabla f(x)$ with bounded variance $\E\|\widetilde g-\nabla f(x)\|^2 \le \sigma^2$. Then inexact ULD-RMM achieves $W_2(\rho_n, \rho) \leq \epsilon$
using
\begin{align}\label{eqn:IULD-RMM-iteration}
\frac{2T}{h}=\widetilde \Theta\Big(\frac{\kappa^{7/6}d^{1/6}}{\epsilon^{1/3}}+\frac{\kappa d^{1/3}}{\epsilon^{2/3}}\Big)
\end{align}
iterations and
\begin{align}\label{eqn:IULD-RMM-query}
b=\frac{d\kappa}{h^3}
\end{align}
queries to the zeroth-order oracle per iteration. The total number of calls is $\frac{bT}{h}$.
\end{lemma}

As a quantum counterpart, we are able to reduce the number of queries from $O(b)$ to $O(1)$ for each iteration in \lem{IULD} and \lem{IULD-RMM} based on \lem{SmoothQuantumGradient}.
Here we are able to choose $\epsilon = O(\frac{\sigma^2}{d^3L})$ to preserve the condition
\begin{align}
\E\|\widetilde g-\nabla f(x)\|^2 \le \E\|\widetilde g-\nabla f(x)\|_1^2 \le \sigma^2
\end{align}
used in \lem{IULD} and \lem{IULD-RMM} with $O(1)$ additional quantum queries. The total number of calls is $O(\frac{T}{h})$ in \lem{IULD} and \lem{IULD-RMM}, the same scaling as in \lem{ULD} and \lem{ULD-RMM}. The query complexities of ULD and ULD-RMM are as follows.

\begin{theorem}\label{thm:quantum-IULD}
Assume the target distribution $\rho$ is strongly log-concave with $L$-smooth and $\mu$-strongly convex negative log-density. Let $\rho_n$ be the distribution of the underdamped Langevin diffusion with the initial point $x_0$ satisfying $\|x_0-x^{\ast}\|\le D$, step size $h\le\frac{\epsilon}{104\kappa}\sqrt{\frac{1}{d/\mu+D^2}}$, and time $T\ge\frac{\kappa}{2}\log\Bigl(\frac{24\sqrt{d/\mu+D^2}}{\epsilon}\Bigr)$. Then quantum inexact ULD (\algo{quantum-IULD}) achieves
\begin{align}
\E\left(\|\widehat{X}_n-X_T\|^2\right) &\leq \widetilde O\Big(\frac{d^2\kappa^2h^2}{\mu}\Big), \\
W_2(\rho_n, \rho) &\leq \epsilon, \label{eqn:quantum-ULD-error}
\end{align}
using
\begin{align}\label{eqn:quantum-ULD-query}
\frac{T}{h}=\widetilde \Theta\Big(\frac{\kappa^2\sqrt{d}}{\epsilon}\Big)
\end{align}
queries to the quantum evaluation oracle.
\end{theorem}
\begin{proof}
  By \lem{IULD}, we know that the number of iterations of ULD with an inexact gradient oracle is $\widetilde{O}(\kappa^2\sqrt{d}/\epsilon)$, as long as the oracle satisfies $\E\|\widetilde{g}-\nabla f(x)\|^2\leq \sigma^2$. By \lem{SmoothQuantumGradient}, this condition can be achieved by the quantum gradient algorithm such that each gradient computation takes  $O(1)$ queries to the quantum evaluation oracle. Therefore, the total number of queries is $\widetilde{O}(\kappa^2\sqrt{d}/\epsilon)$ for the quantum inexact ULD.
\end{proof}

\begin{theorem}\label{thm:quantum-IULD-RMM}
Assume the target distribution $\rho$ is strongly log-concave with $L$-smooth and $\mu$-strongly convex negative log-density. Let $\rho_n$ be the distribution of the randomized midpoint method for underdamped Langevin diffusion with initial point $x_0$, step size $h\le\min\Bigl\{\frac{\epsilon^{1/3}\mu^{1/6}}{\kappa^{1/6}d^{1/6}\log^{1/6}\bigl(\frac{\sqrt{d/\mu}}{\epsilon}\bigr)},\frac{\epsilon^{2/3}\mu^{1/3}}{d^{1/3}\log^{1/3}\bigl(\frac{\sqrt{d/\mu}}{\epsilon}\bigr)}\Bigr\}$, and time $T\ge2\kappa\log\Bigl(\frac{20d/\mu}{\epsilon^2}\Bigr)$. Then quantum inexact ULD-RMM (\algo{quantum-IULD-RMM}) achieves
\begin{align}
\E\left(\|\widehat{X}_n-X_T\|^2\right) &\leq \widetilde O\Big(\frac{d\kappa h^6}{\mu}+\frac{dh^3}{\mu}\Big),  \\
W_2(\rho_n, \rho) &\leq \epsilon, \label{eqn:quantum-ULD-RMM-error}
\end{align}
using
\begin{align}\label{eqn:quantum-ULD-RMM-query}
\frac{2T}{h}=\widetilde \Theta\Big(\frac{\kappa^{7/6}d^{1/6}}{\epsilon^{1/3}}+\frac{\kappa d^{1/3}}{\epsilon^{2/3}}\Big)
\end{align}
queries to the quantum evaluation oracle.
\end{theorem}
The proof is almost the same as \thm{quantum-IULD}, so we omit it here.

\begin{algorithm}[ht]
\KwInput{Function $f$, step size $h$, time $T$, and a sample $x_0$ from a starting distribution $\rho_0$, evaluation error $\epsilon$, Lipschitz constant $L$, smoothness parameter $\beta$}
\KwOutput{Sequence $x_h^h, x_{2h}^h, \ldots, x_{\lfloor T\rfloor + 1}^h$}
Compute $x_0^h\leftarrow x_0$\\
Compute %
$\widetilde g(x_0)\leftarrow\textsf{SmoothQuantumGradient}(f, \epsilon, L, \beta, x_0)$\\
\For{$t=0, h, \ldots, \lfloor T\rfloor$}{
Draw $W_{1,t}^h = \int_0^he^{2(s-h)}\d B_{t+s}$, $W_{2,t}^h = \int_0^h(1-e^{2(s-h)})\d B_{t+s}$\\
Compute
$v_{t+h}^h = e^{-2h}v_t^h + \frac{1}{2L}(1-e^{-2h})\widetilde g(x_t^h) + \frac{2}{\sqrt{L}}W_{1,t}^h$\\
$x_{t+h}^h = x_t^h + \frac{1}{2}(1-e^{-2h})v_t^h + \frac{1}{2L}[h-(1-e^{-2h})]\widetilde g(x_t^h) + \frac{1}{\sqrt{L}}W_{2,t}^h$\\
Compute %
$\widetilde g(x_{t+h}^h)\leftarrow\textsf{SmoothQuantumGradient}(f, \epsilon, L, \beta, x_{t+h}^h)$\\
}
\caption{Quantum Inexact Underdamped Langevin Dynamics (Quantum IULD)}
\label{algo:quantum-IULD}
\end{algorithm}

\begin{algorithm}[ht]
\KwInput{Function $f$, step size $h$, time $T$, and a sample $x_0$ from a starting distribution $\rho_0$}
\KwOutput{Sequence $x_h^h, x_{2h}^h, \ldots, x_{\lfloor T\rfloor + 1}^h$}
Compute $x_0^h\leftarrow x_0$, $y_0^h\leftarrow x_0$\\
Compute %
$\widetilde g(x_0^h)\leftarrow\textsf{SmoothQuantumGradient}(f, \epsilon, L, \beta, x_0^h)$\\
$\widetilde g(y_0^h)\leftarrow\textsf{SmoothQuantumGradient}(f, \epsilon, L, \beta, y_0^h)$\\
\For{$t=0, h, \ldots, \lfloor T\rfloor$}{
Draw $W_{1,t}^h = \int_0^he^{2(s-h)}\d B_{t+s}$, $W_{2,t}^h = \int_0^h(1-e^{2(s-h)})\d B_{t+s}$, $W_{3,t}^h = \int_0^{\alpha h}(1-e^{2(s-h)})\d B_{t+s}$\\
Compute
$v_{t+h}^h = e^{-2h}v_t^h + \frac{h}{L}e^{-2(1-\alpha)h}\widetilde g(y_t^h) + \frac{2}{\sqrt{L}}W_{1,t}^h$\\
$x_{t+h}^h = x_t^h + \frac{1}{2}(1-e^{-2h})v_t^h + \frac{h}{2L}(1-e^{-2(1-\alpha)h})\widetilde g(y_t^h) + \frac{1}{\sqrt{L}}W_{2,t}^h$\\
$y_{t+h}^h = x_t^h + \frac{1}{2}(1-e^{-2\alpha h})v_t^h + \frac{1}{2L}[\alpha h-(1-e^{-2\alpha h})]\widetilde g(x_t^h) + \frac{1}{\sqrt{L}}W_{3,t}^h$\\
Compute %
$\widetilde g(x_{t+h}^h)\leftarrow\textsf{SmoothQuantumGradient}(f, \epsilon, L, \beta, x_{t+h}^h)$\\
$\widetilde g(y_{t+h}^h)\leftarrow\textsf{SmoothQuantumGradient}(f, \epsilon, L, \beta, y_{t+h}^h)$\\
}
\caption{Quantum Inexact Underdamped Langevin Dynamics with Randomized Midpoint Method (Quantum IULD-RMM)}
\label{algo:quantum-IULD-RMM}
\end{algorithm}

\subsection{Quantum MALA}\label{append:q_MALA_sampling}

In \append{classical_MALA}, we introduce several classical results on the mixing of MALA. Then, in \append{quantum_speedup_MALA}, we describe how to implement a quantum walk for MALA and describe the quantum speedup for MALA with a Gaussian initial distribution. Finally, in \append{q_MALA_warm}, we discuss quantum MALA with a warm start distribution.

\subsubsection{Mixing time and spectral gap of MALA}\label{append:classical_MALA}
The Metropolis adjusted Langevin algorithm (MALA) is a key method for sampling log-concave distributions. Classically, the state-of-the-art mixing time bound of MALA was proven by \citet{lst20}. They show that MALA is equivalent to the Metropolized Hamiltonian Monte Carlo method (HMC). Then, they consider the following Metropolized HMC algorithm (\algo{HMC}) and use the blocking conductance analysis of \citet{klm06} to upper bound the mixing time.

\begin{algorithm}[htbp]
  \KwInput{Initial point $x_0 \in \R^d$, step size $\eta$}
  \KwOutput{Sequence $\{x_k\}, k\geq 0$}
  \For{$k\geq 0$}{
    Draw $v_k\sim {\cal N}(0,I_d)$\\
    $(\wt{x}_k,\wt{v}_k)\gets \textsc{LeapFrog}(\eta, x_k,v_k)$\\
    Draw $u\sim {\cal U}([0,1])$\\
    \If{$u\leq \min\{1,\exp({\cal H}(x_k,v_k)-{\cal H}(\wt{x}_k,\wt{v}_k))\}$}{
      $x_{k+1}\gets \wt{x}_k$
    }
    \Else{
      $x_{k+1}\gets x_k$
    }
  }
  \caption{Metropolized HMC: \textsc{HMC}($x_0, \eta$)}
  \label{algo:HMC}
\end{algorithm}
\begin{algorithm}[htbp]
  \KwInput{Points $x,v \in \R^d$, step size $\eta$}
  \KwOutput{Points $\wt{x},\wt{v}\in \R^d$}
  $v'\gets v-\frac{\eta}{2}\nabla f(x)$\\
  $\wt{x}\gets x + \eta v'$\\
  $\wt{v} \gets v' - \frac{\eta}{2}\nabla f(x)$
  \caption{\textsc{LeapFrog}($\eta, x, v$)}
  \label{algo:leapfrog}
\end{algorithm}

Define ${\cal H}(x,v):=f(x) + \frac{1}{2}\|v\|_2^2$. Let $\d \pi^\star$ denote the target distribution, i.e., $\d \pi^\star(x)/\d x \propto \exp(-f(x))$. Then the Markov chain defined by \algo{HMC} has the following property.

\begin{lemma}[\citet{lst20}]\label{lem:HMC_stationary}
The Markov chain of \algo{HMC} is reversible, and its stationary distribution is $\d \pi^\star$.
\end{lemma}

The main result of \citet{lst20} is the following theorem on the mixing time of \algo{HMC}.

\begin{theorem}[Mixing of Hamiltonian Monte Carlo, Theorem 4.7 of \citet{lst20}]\label{thm:HMC_mixing}
There is an algorithm initialized from a point drawn from ${\cal N}(x^\star, L^{-1}I_d)$ that iterates \algo{HMC}
\begin{align}
  O\left( \kappa d \log(\kappa / \epsilon) \log(d\log (\kappa/\epsilon))\log(1/\epsilon) \right)
\end{align}
times and produces a point from a distribution $\rho$ such that $\|\rho - \pi^\star\|_{TV}\leq \epsilon$.
\end{theorem}

The algorithm in the above theorem defines a new Markov chain where in each step, we draw an integer $j$ uniformly from 0 to $O(\kappa d \log(\kappa / \epsilon) \log(d\log (\kappa/\epsilon)))$ and run \algo{HMC} for $j$ iterations. One step of this Markov chain gives a distribution with TV-distance from $\pi^\star$ at most $(2e)^{-1}$ \citet{lst20}. Hence, if we run for $\log(1/\epsilon)$ steps, we get $\epsilon$ TV-distance.

We also use a well-known relation between the mixing time and the spectral gap of Markov chain (see e.g. Chapter 12 in \citet{lp17}).

\begin{theorem}[Mixing time and spectral gap]\label{thm:mixing_spectral}
For a reversible, irreducible, and aperiodic Markov chain, let $t_{\text{mix}}(\epsilon)$ denote the $\epsilon$-mixing time and let $\delta$ denote the spectral gap. Then
\begin{align}
  t_{\text{mix}}(\epsilon) \geq (\delta^{-1} - 1)\log(1/(2\epsilon)).
\end{align}
\end{theorem}

It is easy to verify that the Markov chain of \algo{HMC} is reversible, irreducible, and aperiodic. Hence, together with \thm{HMC_mixing}, we know that the spectral gap of the Markov chain satisfies
\begin{align}\label{eq:MALA_spectral_gap}
  \delta^{-1} \leq O(\kappa d \log(\kappa /\epsilon)\log(d\log(\kappa/\epsilon))).
\end{align}

Furthermore, we can show that MALA converges faster under a certain warm start condition \citep{DCWY18,WSC21,cla21}. We say the initial distribution $\rho_0$ is \emph{$\beta$-warm} if there is a constant $\beta$ independent of $\kappa, d$ such that
\begin{align}\label{eqn:M-warm}
\sup_{S\in \mathcal{B}(\R^d)} \frac{\rho_n(S)}{\rho(S)} \le \beta.
\end{align}
The warmness of the Gaussian $\rho_0=\mathcal{N}(x^{\ast},\frac{1}{L}I)$ satisfies $\beta\le\kappa^{d/2}$ \citep{DCWY18,WSC21}.

Given a $\beta$-warm initial distribution, MALA has the following improved convergence.

\begin{lemma}[{Theorem 1 of \citet{WSC21}}]\label{lem:MALA}
Assume the target distribution $\rho$ is strongly log-concave with $L$-smooth and $\mu$-strongly convex negative log-density. Let $\rho_n$ be the distribution of the $\frac{1}{2}$-lazy version of MALA with $\beta$-warm initial distribution $\rho_0$ and step size $h=c_0(Ld\log^2(\max\{\kappa, d, \frac{\beta}{\epsilon}, c_2\}))^{-1}$. There exist universal constants $c_0, c_1, c_2>0$, such that MALA achieves
\begin{align}\label{eqn:MALA-error}
d_{TV}(\rho_n, \rho) \le \epsilon
\end{align}
after
\begin{align}\label{eqn:MALA-step}
n \ge c_1\kappa\sqrt{d}\log^3(\max\{\kappa, d, \frac{\beta}{\epsilon}, c_2\})
\end{align}
steps.
\end{lemma}

\subsubsection{Quantum walks for MALA}\label{append:quantum_speedup_MALA}
The goal of this section is to show quantum speedup for the Metropolis adjusted Langevin algorithm (MALA) using the continuous-space quantum walks defined by \citet{cch19}, which generalize the discrete-time quantum walk of \citet{szegedy2004quantum} to continuous space.

Given a transition density function $p$, the quantum walk is characterized by the states
\begin{align}
|\phi_{x}\>:=|x\>\otimes\int_{\Omega}\d y\,\sqrt{p_{x\to y}}|y\>\qquad\forall x\in\R^{n},
\end{align}
where $p_{x\to y}:=p(x,y)$.

Now, denote
\begin{align}
U:=\int_{\Omega}\d x\,|\phi_{x}\>(\<x|\otimes\<0|),\quad \Pi:=\int_{\Omega}\d x\,|\phi_{x}\>\<\phi_{x}|,\quad S:=\int_{\Omega}\int_{\Omega}\d x\,\d y\,|x,y\>\<y,x|.
\end{align}
A single step of the quantum walk is defined as the unitary operator
\begin{align}\label{eqn:walk-operator}
W:=S(2\Pi-I).
\end{align}

The following theorem characterizes the eigenvalues of the quantum walk operator.

\begin{theorem}[Theorem 3.1 of \citet{cch19}]\label{thm:quantum-walk-main}
Let
\begin{align}\label{eqn:discriminant}
D:=\int_{\Omega}\int_{\Omega}\d x\,\d y\,\sqrt{p_{x\to y}p_{y\to x}}|x\>\<y|
\end{align}
denote the \emph{discriminant operator} of $p$. Let $\Lambda$ be the set of eigenvalues of $D$, so that $D=\int_{\Lambda}\d\lambda\,\lambda|\lambda\>\<\lambda|$. Then the eigenvalues of the quantum walk operator $W$ in \eqn{walk-operator} are $\pm 1$ and $\lambda\pm i\sqrt{1-\lambda^{2}}$ for all $\lambda\in\Lambda$.
\end{theorem}

Furthermore, Ref.~\citet{cch19} shows that, for a reversible Markov chain with unique stationary distribution $\rho$, the state
\begin{align}
  \ket{\rho_W} := \int_\Omega \d x \sqrt{\rho_x} \ket{\phi_x}
\end{align}
is the unique eigenvalue-1 eigenstate of the quantum walk operator $W$ restricted to the subspace $\mathrm{span}_{\lambda \in \Lambda}\{T\ket{\lambda}, ST\ket{\lambda}\}$. %
Hence, the stationary distribution corresponds to an eigenstate with eigenphase $0$, while the other eigenstates have eigenphase at least $\sqrt{2\delta}$, where $\delta$ is the spectral gap of $P$. Thus, by the quantum phase estimation algorithm with $O(1/\sqrt{\delta})$ calls to $W$, we can distinguish the stationary state from other eigenstates, achieving quadratic speedup over the classical mixing time $O(1/\delta)$.

To implement a quantum version of \algo{HMC}, we prepare the initial state $\ket{\pi_0}$ and implement the quantum walk operator $W$.

\paragraph{Initial state.}
For the initial state $\ket{\rho_0}$, by \thm{HMC_mixing}, it suffices to take $\rho_0= {\cal N}(x^\star, L^{-1} I_d)$, where $x^\star$ is the minimum point of $f(x)$. Suppose we already have $x^\star$. Appendix A.3 of \citet{cch19} shows that the state
\begin{align}
  \int_{\R^d} \left(\frac{L}{2\pi}\right)^{d/4} e^{-\frac{L}{4}\|z\|_2^2}\ket{z} \d z
\end{align}
can be efficiently prepared by applying a Box-Muller transformation to the state corresponding to the uniform distribution (i.e., an equal superposition of points). %
Then, for the $i$th register, we apply the shift operation $U_{\mathrm{shift}}$ with $U_{\mathrm{shift}}\ket{x_i} = \ket{x_i + x^\star_i}$. The resulting state is
\begin{align}
  \ket{\rho_0} = \int_{\R^d} \left(\frac{L}{2\pi}\right)^{d/4} e^{-\frac{L}{4}\|z-x^\star\|_2^2}\ket{z} \d z.
\end{align}

\paragraph{Quantum walk operator.}
The quantum walk operator $W(P)$ can be implemented\footnote{As shown in \citet{wocjan2008speedup}, $W(P)=U^\dagger S U R U^\dagger S U R$, where $S$ is the swap gate and $R$ is a reflection operator with respect to the state space $\mathrm{span}\{\ket{x}\ket{0}:x\in \R^d\}$.} using the quantum walk update unitary $U$ that maps each point $\ket{x}$ to the superposition $\int_{\R^d} \d y \sqrt{p_{x\to y}}\ket{y}$. We show how to efficiently implement $U$.

We first use $d$ ancilla registers to prepare a standard Gaussian state
\begin{align}
  \ket{v}:=\int_{\R^d} \left(\frac{L}{2\pi}\right)^{d/4} e^{-\frac{L}{4}\|z\|_2^2}\ket{z} \d z.
\end{align}
Then, we implement a unitary $U_{\mathrm{LF}}$ defined by \textsc{LeapFrog}($\eta, x, v$) (\algo{leapfrog}) such that for two points $x,v\in \R^d$, $U_{\mathrm{LF}}\ket{x,v}\ket{0} = \ket{x,v}\ket{\wt{x}, \wt{v}}$, by querying the gradient oracle of $f$ twice. Then we use another ancilla register to prepare the state
\begin{align}
  \ket{u}:=\int_{[0,1]}\d \ket{z}.
\end{align}
Based on the registers $x, v, \wt{x}, \wt{v}, u$, we can decide using two queries to the evaluation oracle for $f$ whether the target register $y$ should be $x$ or $\wt{x}$. Overall, this process implements the following mapping (up to a hidden normalization factor):
\begin{align}
  \ket{x}\ket{0}\mapsto \ket{x}\int_{\R^d} \int_{[0,1]} \d v \d u \exp(-\|v\|_2^2/4) \ket{v}\ket{\wt{x}}\ket{\wt{v}}\ket{u}\ket{y}.
\end{align}
Finally, uncomputing the $v, \wt{x}, \wt{v}$ registers using $U_{\mathrm{LF}}^{\dagger}$ and the inverse of the unitary preparing $\ket{v}$, we obtain a superposition of points with the correct transition density.

\begin{algorithm}[htbp]
  \KwInput{Quantum state $\ket{\rho}=\int_{\R^d}\d x \sqrt{\rho_x}\ket{x}$.}
  \KwOutput{Quantum state $\ket{\phi}=\int_{\R^d}\d x \sqrt{\rho_x}\ket{x}\int_{\R^d} \d y \sqrt{p_{x\rightarrow y}}\ket{y}$.}

  Prepare $\ket{\rho}\ket{v}$ where $\ket{v}$ is a $d$-dimensional Gaussian state\\
  Apply the leap-frog process to $\ket{\rho}\ket{v}$: $\int_{\R^d}\d x \, \d v  \ket{x,v}\ket{\wt{x},\wt{v}}$\tcc*{Query ${\cal O}_{\nabla f}$ twice}
  Prepare $\int_{\R^d}\d x \, \d v  \ket{x,v}\ket{\wt{x},\wt{v}} \int_{[0,1]}\d u \ket{u}$\\
  Compute the target point $y$: $\int_{\R^d}\d x \, \d v\int_{[0,1]}\d u   \ket{x,v}\ket{\wt{x},\wt{v}}\ket{u}\ket{y}$\tcc*{Query ${\cal O}_{f}$ twice}
  Uncompute the $v,\wt{x},\wt{v},u$ registers: $\int \d x \sqrt{\rho_x}\ket{x}\int_{\R^d} \d y \sqrt{p_{x\rightarrow y}}\ket{y}$
  \caption{\textsc{QuantumUpdateUnitary}}
  \label{algo:q_walk_one_step}
\end{algorithm}

\begin{algorithm}[htbp]
  \KwInput{Evaluation oracle ${\cal O}_f$, gradient oracle ${\cal O}_{\nabla f}$, initial state $\ket{\rho_0}$}
  \KwOutput{Quantum state $\ket{\widetilde{\rho}}$ close to the stationary distribution state $\int_{\R^d} e^{-f(x)}\d \ket{x}$}
  Construct quantum walk update unitary $U$ from \textsc{QuantumUpdateUnitary} (\algo{q_walk_one_step}) with ${\cal O}_f$ and ${\cal O}_{\nabla f}$\\
  Implement the quantum walk operator $W(P)$\\
  Perform $\frac{\pi}{3}$-amplitude amplification with $W(P)$ on the state $\ket{\rho_0}\ket{0}$\\
  \Return the resulting state $\ket{\widetilde{\rho}}$
  \caption{\textsc{QuantumMALA}}
  \label{algo:q_walk_MALA}
\end{algorithm}

\begin{lemma}[Continuous-space quantum walk implementation]\label{lem:continuous_walk}
The Markov chain of \algo{HMC} can be implemented as a continuous-space quantum walk where the quantum walk unitary for one step can be implemented with 2 queries to the gradient oracle, 2 queries to the evaluation oracle, and $O(d)$ quantum gates.
\end{lemma}

If we have a sequence of slowly varying log-concave distributions $\rho_0,\dots,\rho_r$, we can quantumly sample from $\rho_r$ via MALA with a quadratic speedup.

\begin{theorem}[Quantum speedup for slowly varying Markov chains \citet{wocjan2008speedup}]\label{thm:quantum_slow_varying}
  Let $M_0, \dots , M_r$ be classical reversible Markov chains with stationary distributions $\rho_0, \dots , \rho_r$ such that each chain has spectral gap at least $\delta^{-1}$. Assume that $|\bra{\rho_i}\rho_{i+1}\rangle|\geq p$ for some $p > 0$ and all $i\in \{0, \dots, r - 1\}$, and that we can prepare the state $\ket{\rho_0}$. Then, for any $\epsilon > 0$, there is a quantum algorithm which produces a quantum state $\ket{\widetilde{\rho_r}}$ such
  that $\|\ket{\widetilde{\rho}_r} - \ket{\rho_r}\ket{0^a}\|\leq \epsilon$, for some integer $a$. The algorithm uses
  \begin{align}
    \tilde{O}\left(\delta^{-1/2}\cdot \frac{r}{p}\right)
  \end{align}
  applications of the quantum walk operators $W_i$ corresponding to the chains $M_i$ for $i\in [r]$.
\end{theorem}

Combining \thm{quantum_slow_varying} with the lower bound for the spectral gap of MALA (Eq.~\eqref{eq:MALA_spectral_gap}), we find the following.

\begin{corollary}[Quantum speedup for slowly varying MALAs]\label{cor:MALA_slow_vary}
  Let $\rho_0,\dots, \rho_r$ be a sequence of log-concave distributions such that $|\langle \rho_i|\rho_{i+1}\rangle|\geq p$ for some $p>0$ and for all $i\in \{0,\dots,r-1\}$. Suppose we can prepare the initial state $\ket{\rho_0}$. Then, for any $\epsilon>0$, there is a quantum procedure to produce a state $\ket{\widetilde{\rho}_r}$ such that $\|\ket{\widetilde{\rho}_r} - \ket{\rho_r}\|\leq \epsilon$ using
  \begin{align}
    \widetilde{O}(\kappa^{1/2}d^{1/2} \cdot (r/p))
  \end{align}
  applications of the quantum walk operators $W_i$ corresponding to the MALA procedure for $\rho_i$ for $i\in [r]$.
\end{corollary}

\begin{algorithm}[htbp]
  \KwInput{Evaluation oracle ${\cal O}_f$, gradient oracle ${\cal O}_{\nabla f}$, smoothness parameter $L$, convexity parameter $\mu$}
  \KwOutput{Quantum state $\ket{\widetilde{\rho}}$ close to the stationary distribution state $\int_{\R^d} e^{-f(x)}\d \ket{x}$}
  Compute the cooling schedule parameters $\sigma_1,\dots, \sigma_M$\\
  Prepare the state $\ket{\rho_0}\propto \int_{\R^d} e^{-\frac{1}{4}\|x\|^2/\sigma_1^2}\d \ket{x}$\\
  \For{$i\gets 1,\dots, M$}{
    Construct ${\cal O}_{f_i}$ and ${\cal O}_{\nabla f_i}$ where $f_i(x)=f(x)+\frac{1}{2}\|x\|^2/\sigma_i^2$\\
    $\ket{\rho_i}\gets \textsc{QuantumMALA}({\cal O}_{f_i}, {\cal O}_{\nabla f_i}, \ket{\rho_{i-1}})$ (\algo{q_walk_MALA})
  }
  \Return $\ket{\rho_M}$
  \caption{\textsc{QuantumMALAforLog-concaveSampling}}
  \label{algo:q_MALA_sampling}
\end{algorithm}

\begin{theorem}[Quantum MALA for log-concave sampling]\label{thm:q_MALA_sampling}
  Assume the target distribution $\rho\propto e^{-f}$ is strongly log-concave with $f\colon\R^d\rightarrow \R_+$ being $L$-smooth and $\mu$-
  strongly convex.  Let $\ket{\rho}$ be the quantum state corresponding to the distribution $\rho$. Then, for any $\epsilon>0$, there is a quantum algorithm (\algo{q_MALA_sampling}) that prepares a state $\ket{\widetilde{\rho}}$ such that $\|\ket{\widetilde{\rho}} - \ket{\rho}\|\leq \epsilon$ using $\widetilde{O}(\kappa^{1/2}d)$ queries to the evaluation oracle ${\cal O}_f$ and gradient oracle ${\cal O}_{\nabla f}$.
\end{theorem}
\begin{proof}
  By \lem{slowly_varying}, we know that the cooling schedule $\sigma_1,\dots,\sigma_M$ gives a sequence of slowly-varying Markov chains with overlap $p=\Omega(1)$. We also know that the length of the schedule is $M=\tilde{O}(d^{1/2})$.

  Hence, by \cor{MALA_slow_vary} with $r=\tilde{O}(d^{1/2})$ and $p=\Omega(1)$, $\ket{\widetilde{\rho}}$ can be prepared by $\tilde{O}(\kappa^{1/2}d)$ quantum walk steps. By \lem{continuous_walk}, each step queries ${\cal O}_f$ and ${\cal O}_{\nabla f}$ twice. The result follows.
\end{proof}

\subsubsection{Quantum MALA with a warm start}\label{append:q_MALA_warm}
We now show that quantum MALA can achieve further speedup in query complexity when the initial distribution is a warm start (\thm{q_MALA_warm}).

The following proposition shows that when the initial distribution is warm, then the eigenvalues whose eigenspaces have large overlap with the initial state are bounded away from 1. Our result improves Proposition 4.2 of \citet{chakrabarti2018quantum}, which required the initial distribution to satisfy an $L_2$-norm condition, by instead relying only on the standard warmness of the initial distribution.

\begin{lemma}[Effective spectral gap for warm start]\label{lem:effect_spectral_gap}
Let $M = (\Omega, p)$ be an ergodic reversible Markov chain with a transition operator
$P$ and unique stationary state with a corresponding density $\rho$. Let $\{(\lambda_i, f_i)\}$ be the set
of eigenvalues and eigenfunctions of $P$, and let $\ket{u_i}$ be the eigenvectors of the corresponding quantum walk operator $W$. Let $\rho_0$ be a probability density that is a warm start for $\rho$ and mixes up to TV-distance $\epsilon$ in $t$ steps of $M$. Furthermore, assume that $\rho_0$ is a $\beta$-warm start of $\rho$. Let $\ket{\phi_{\rho_0}}$ be the state obtained by applying the quantum walk update operator $U$ to the state $\ket{\rho_0}$:
\begin{align}
  \ket{\phi_{\rho_0}} = \int_\Omega \sqrt{\rho_0(x)}\int_\Omega \sqrt{p_{x\rightarrow y}}\ket{x}\ket{y} \, \d x \,\d y.
\end{align}
Then $|\langle \phi_{\rho_0}|u_i\rangle| = O(\beta\sqrt{\epsilon})$ for all $i$ with $1>\lambda_i \geq 1-O(1/t)$.
\end{lemma}
\begin{proof}
  Let $S:=\bigl\{x\in \R^d: \frac{\rho(x)}{\rho_0(x)}\geq \frac{1}{\epsilon}\bigr\}$. Since $\E_{\rho_0}\left[\frac{\rho(x)}{\rho_0(x)}\right] = 1$, by Markov's inequality, we have
  \begin{align}
    \int_{S} \rho_0(x)\d x = \Pr_{\rho_0}[x\in S] \leq \epsilon.
  \end{align}
  Then we define a quantum state $\ket{\rho_1}$ such that $\langle\rho_1 | x\rangle = \langle \rho_0|x\rangle$ for $x\notin S$, and $\langle \rho_1|x\rangle=0$ for $x\in S$. Furthermore, let $\ket{\phi_{\rho_1}}:=U\ket{\rho_1}$.

  We have
  \begin{align}
    \|\ket{\phi_{\rho_0}} - \ket{\phi_{\rho_1}} \| = \left\|\int_{S}\sqrt{\rho_0(x)} T\ket{x}\d x\right\| = \left|\int_S \rho_0(x)\d x\right|^{1/2} \leq  \sqrt{\epsilon},
  \end{align}
  where $T$ is the isometry
  \begin{align*}
    T:=\int_{\Omega} \int_{\Omega} \sqrt{p_{x\rightarrow y}}\ket{x,y}\bra{x}\, \d x \, \d y.
  \end{align*}

  Moreover, by Eqs.~(4.35) and (4.36) in \citet{cch19}, if $1>\lambda_i\geq 1-O(1/t)$, we have
  \begin{align}
    |\langle \phi_{\rho_1}|u_i\rangle| \leq &~ 2 |\langle \rho_1 | v_i\rangle|=2\left|\int_{\overline{S}} \sqrt{\frac{\rho(x)}{\rho_0(x)}}\frac{\rho_0(x)f_i(x)}{\rho(x)} \d x\right|\leq\frac{2\langle \rho_0, f_i\rangle_{\rho}}{\epsilon^{1/2}}=O(\beta \sqrt{\epsilon}),
  \end{align}
  where the third step follows from $\frac{\rho(x)}{\rho_0(x)}\leq \frac{1}{\epsilon}$ for $x\notin S$ and the Cauchy-Schwarz inequality, %
  and the last step follows from the claim that $\langle \rho_0, f_i\rangle_\rho=O(\beta\epsilon)$.

Combining these observations, we find
  \begin{align}
    |\langle \phi_{\rho_0}|u_i\rangle| \leq |\langle \phi_{\rho_0} - \phi_{\rho_1}|u_i\rangle| + |\langle \phi_{\rho_1}|u_i\rangle| \leq \sqrt{\epsilon} + O(\beta \sqrt{\epsilon}) = O(\beta\sqrt{\epsilon})
  \end{align}
  when $1>\lambda_i \geq 1-O(1/t)$, which gives the desired result.

  Now, it remains to prove the claim that $\langle \rho_0, f_i\rangle_\rho=O(\beta\epsilon)$ for $1>\lambda_i > 1 - O(1/t)$. Suppose $\rho_0$ can be decomposed in the eigenbasis of $P$ as $\rho_0=\rho+\sum_{i=2}^\infty \langle \rho_0, f_i\rangle_\rho f_i$. Then $P^t\rho_0 = \rho + \sum_{i=2}^\infty \lambda_i^t \langle \rho_0, f_i\rangle_\rho f_i$, where $\lambda_i$ is the eigenvalue of $f_i$. Since $\|P^t \rho_0 - \rho\|_1\leq \epsilon$, by \fac{tv_implies_chi_square}, we have $\|P^t\rho_0 - \rho\|_\rho \leq \beta \epsilon$. Hence, by the orthogonality of $f_i$, we have $\lambda_i^t \langle a, f_i\rangle_\rho\leq \beta\epsilon$. Therefore, when $1>\lambda_i > 1 - O(1/t)$, we have $\lambda_i^t=\Omega(1)$, which implies that $\langle \rho_0, f_i\rangle_\rho= O(\beta \epsilon)$.
\end{proof}

\begin{fact}\label{fac:tv_implies_chi_square}
  Let $\rho_0$ be a $\beta$-warm start of a Markov chain with transition operator $P$ and stationary distribution $\rho$. If $\|P^t\rho_0 - \rho\|_{1}\leq \epsilon$ for some $t>0$, then $\|P^t\rho_0 - \rho\|_{\rho}\leq \beta \cdot \epsilon$.
\end{fact}
\begin{proof}
  We expand $\|P^t\rho_0 - \rho\|_{\rho}$ in terms of its definition, giving
  \begin{align}
    \|P^t\rho_0 - \rho\|_{\rho} = \int_{\R^d} \frac{(P^t\rho_0 - \rho)^2(x)}{\rho(x)}\d x = \int_{\R^d} \frac{(P^t\rho_0)^2(x)}{\rho(x)}\d x - 1 = \chi^2(P^t\rho_0, \rho),
  \end{align}
  where the second step follows since $P^t\rho_0$ and $\rho$ are distributions, and the last step follows from the definition of $\chi^2$-distance. Let $\rho_1:=P^t\rho_0$. By Lemma 27 in \citet{cla21}, we know that $\rho_1$ is also $\beta$-warm. Furthermore, by Lemma 28 in \citet{cla21}, the $\chi$-square distance can be upper bounded by
  \begin{align}
    \chi^2(P^t\rho_0, \rho) \leq \beta \cdot \|P^t\rho_0 - \rho\|_{1}\leq \beta\epsilon
  \end{align}
as claimed.
\end{proof}

\begin{theorem}[Quantum MALA with warm start]\label{thm:q_MALA_warm}
Let $\ket{\rho_0}$ be a $\beta$-warm start with respect to the log-concave distribution $\rho\propto e^{-f}$. Let ${\cal T}$ be the mixing time of classical MALA with initial distribution $\rho_0$, i.e., $\|P^{\cal T}\rho_0 - \rho\|_{\mathrm{TV}}\leq \epsilon$. Then there is a quantum algorithm that prepares a state $\ket{\widetilde{\rho}}$ that is $\epsilon$-close to $\ket{\rho}$ using
  \begin{align}
    \widetilde{O}\left(\sqrt{\cal T\log(\beta)}\right)
  \end{align}
  queries to the evaluation oracle ${\cal O}_f$ and gradient oracle ${\cal O}_{\nabla f}$.
\end{theorem}
\begin{proof}
  Let $\ket{\rho_0}$ be the initial state corresponding to a distribution $\rho_0$ that is $\beta$-warm. Then we know that $t=\widetilde{O}({\cal T}\log(\beta))$ steps of the classical MALA random walk suffice to achieve $\|P^t\rho_0 -\rho\|_1\leq \epsilon/\beta^2$.

  By \lem{effect_spectral_gap},  we have $\ket{\rho_0} = \ket{\rho_0'}+\ket{e}$, where $\ket{\rho_0'}$ lies in the space of eigenvectors $\ket{v_i}$ of $W$ such that $\lambda_i=1$ or $\lambda_i\leq 1-\widetilde{\Omega}({\cal T}^{-1}\log^{-1}(\beta))$, and $\|\ket{e}\|\leq \epsilon_1$.

  Hence, by Corollary 4.1 in \citet{chakrabarti2018quantum}, the approximate reflection in the quantum walk $\widetilde{R}$ can be implemented using $\widetilde{O}({\cal T}^{1/2}\log^{1/2}(\beta))$ calls to the controlled-$W$ operator. Furthermore, the approximated stationary state $\ket{\widetilde{\rho}}$ can be prepared via $O(\log(1/\epsilon_2))$ recursive levels of $\frac{\pi}{3}$-amplitude amplification such that $\|\ket{\rho}-\ket{\widetilde{\rho}}\|\leq \epsilon_2$.

  By choosing $\epsilon_1=\epsilon/(2\log(2/\epsilon))$ and $\epsilon_2=\epsilon/2$, we achieve a final approximation error of $O(\epsilon_1\log(\epsilon_2)+\epsilon_2)\leq \epsilon$. By \lem{continuous_walk}, each controlled-$W$ operator takes a constant number of queries to ${\cal O}_f$ and ${\cal O}_{\nabla f}$, which gives the desired result.
\end{proof}

\begin{remark}\label{rmk:warm_start_sampling}
  The above theorem is quite general and works for both polynomially-warm and Gaussian initial distributions.
  \begin{itemize}
    \item When $\rho_0$ is a $\beta$-warm start for $\beta\leq \poly(d)$\footnote{We assume that we can access an arbitrary number of copies of the quantum state $\ket{\rho_0}$.}, by \lem{MALA}, we have ${\cal T}=\widetilde{O}(\kappa d^{1/2})$. Then the query complexity of quantum MALA is $\widetilde{O}(\kappa^{1/2}d^{1/4})$.
    \item When $\rho_0={\cal N}(0, L^{-1}I_d)$, the warmness is $\kappa^{d/2}$, and by \thm{HMC_mixing}, we have ${\cal T}=\widetilde{O}(\kappa d)$. Hence, the query complexity of quantum MALA is $\widetilde{O}(\kappa^{1/2}d)$, recovering \thm{q_MALA_sampling}.
  \end{itemize}
\end{remark}

\begin{algorithm}[htbp]
  \KwInput{Evaluation oracle ${\cal O}_f$, gradient oracle ${\cal O}_{\nabla f}$, smoothness parameter $L$, convexity parameter $\mu$, warm-start state $\ket{\rho_0}$}
  \KwOutput{Quantum state $\ket{\widetilde{\rho}}$ close to the stationary distribution state $\int_{\R^d} e^{-f(x)}\d \ket{x}$}
  $\ket{\widetilde{\rho}}\gets \textsc{QuantumMALA}({\cal O}_{f_i}, {\cal O}_{\nabla f_i}, \ket{\rho_{0}})$ (\algo{q_walk_MALA})\\
  \Return $\ket{\widetilde{\rho}}$
  \caption{\textsc{QuantumMALAwithWarmStart}}
  \label{algo:q_MALA_sampling_warm}
\end{algorithm}

\section{Quantum Algorithm for Estimating Normalizing Constants: Details}\label{append:estimating_nc_full}
We now come back to the problem of estimating the normalizing constant.
\subsection{Quantum MALA and annealing}\label{sec:quantum-MALA}

In this section, we first describe a quantum speedup for the annealing process via a quantum-accelerated Monte Carlo method, which quadratically improves the $\epsilon$-dependence of the sampling complexity of the classical algorithm. Then we further reduce the $\kappa$- and $d$-dependence of the query complexity using the quantum MALA procedure developed in \append{q_MALA_sampling}.

\subsubsection{Quantum speedup for the standard annealing process}
Reference \citet{montanaro2015quantum} developed a quantum-accelerated Monte Carlo method for mean estimation with $B$-bounded relative variance.\footnote{Reference \citet{AHN21} improves the scaling from $\widetilde O(B/\epsilon)$ to $\widetilde O(\sqrt{B}/\epsilon)$. Such a result also follows from the quantum Chebyshev inequality of \citet{hamoudi2019Chebyshev}. Since $B=O(1)$ in our case, we apply the original algorithm of \citet{montanaro2015quantum}.} We state the result as follows.

\begin{lemma}[{Theorem 6 of \citet{montanaro2015quantum}}]\label{lem:quantum-relative-variance}
Assume there is an algorithm $\mathcal{A}$ such that $v(\mathcal{A})\ge0$ and $\frac{\mathrm{Var}(v(\mathcal{A}))}{\E[v(\mathcal{A})]^2}\le B$ for some $B\ge1$, and an accuracy $\epsilon<27B/4$. Then there is a quantum algorithm which outputs an estimate $\widetilde \mu$ such that
\begin{align}\label{eqn:quantum-relative-variance}
\mathrm{Pr}\Bigl[|\widetilde \mu - \E[v(\mathcal{A})]| \ge \epsilon \E[v(\mathcal{A})]\Bigr] \le \frac{1}{4},
\end{align}
with
\begin{align}
O\Bigl(\frac{B}{\epsilon}\log^{3/2}\Bigl(\frac{B}{\epsilon}\Bigr)\log\log\Bigl(\frac{B}{\epsilon}\Bigr)\Bigr)
\end{align}
queries to $\mathcal{A}$.
\end{lemma}

\begin{lemma}[{Lemma 1 of \citet{montanaro2015quantum}}]\label{lem:powering-lemma}
Let $\mathcal{A}$ be a (classical or quantum) algorithm which aims to estimate some quantity $\mu$, and whose output $\tilde{\mu}$ satisfies $|\mu - \tilde{\mu}|\leq \epsilon$ except with probability $\gamma$, for some fixed $\gamma<1/2$. Then, for any $\delta>0$, it suffices to repeat $\mathcal{A}$ $O(\log 1/\delta)$ times and take the median to obtain an estimate which is accurate to within $\epsilon$ with probability at least $1 - \delta$.
\end{lemma}

This result provides a way of estimating the telescoping product \eqn{telescoping-Z}. The following theorem and its proof closely follows Theorem 8 of \citet{montanaro2015quantum}, while the definitions of the partition function and the cooling schedule are different.

\begin{theorem}[Quantum speedup of annealing]\label{thm:quantum-annealing}
Let $Z$ be the normalizing constant in \prb{log-concave}. Consider a sequence of values $g_i$ as in \eqn{distribution-g}, with $\frac{\E_{\rho_i}(g_i^2)}{\E_{\rho_i}(g_i)^2} = O(1)$. Further assume that we have the ability to sample $\rho_i$ for $i\in\range{M}$. Then there is a quantum algorithm which outputs an estimate $\widetilde Z$, such that
\begin{align}\label{eqn:quantum-annealing-prob}
\mathrm{Pr}[(1-\epsilon)Z \le \widetilde Z \le (1+\epsilon)Z] \ge \frac{3}{4},
\end{align}
using
\begin{align}\label{eqn:quantum-annealing-sample}
O\Bigl(\frac{M^2}{\epsilon}\log^{3/2}\Bigl(\frac{M}{\epsilon}\Bigr)\log\log\Bigl(\frac{M}{\epsilon}\Bigr)\Bigr) = \widetilde O\Big(\frac{M^2}{\epsilon}\Big)
\end{align}
samples in total.
\end{theorem}

\begin{proof}
For $i\in\range{M}$, we estimate $\E_{\rho_i}(g_i)$ with output $\widetilde g_i$ up to additive error $(\epsilon/2M)\E_{\rho_i}(g_i)$ with failure probability $1/4M$. We output as a final estimate
\begin{align}
\widetilde Z = \widetilde Z_1\prod_{i=1}^{M}\widetilde g_i,
\end{align}
where $\widetilde Z_1$ is the normalizing constant of the Gaussian distribution with variance $\sigma_1^2$ as in \lem{telescoping-1}.  Assuming that all the estimates are indeed accurate, we have
\begin{align}
1-\epsilon \le (1-\frac{\epsilon}{2})\Big(1-\frac{\epsilon}{2M}\Big)^M \le \frac{\widetilde Z}{Z} \le \Big(1+\frac{\epsilon}{2M}\Big)^M \le e^{\epsilon/2} \le 1+\epsilon.
\end{align}
Thus $|\widetilde Z-Z| \le \epsilon Z$ with probability at least
\begin{align}
\Big(1-\frac{1}{4M}\Big)^M \ge 1-\frac{1}{4} = \frac{3}{4}.
\end{align}

Based on \lem{telescoping-M} and \lem{telescoping}, $\frac{\E_{\rho_i}(g_i^2)}{\E_{\rho_i}(g_i)^2}=O(1)$, so $\frac{\mathrm{Var}_{\rho_i}(g_i)}{\E_{\rho_i}(g_i)^2}=O(1)$. By \lem{quantum-relative-variance}, each requires
\begin{align}
O\Bigl(\frac{M}{\epsilon}\log^{3/2}\Bigl(\frac{M}{\epsilon}\Bigr)\log\log\Bigl(\frac{M}{\epsilon}\Bigr)\Bigr)
\end{align}
samples from $\rho_i$, and the total number of samples is
\begin{align}
O\Bigl(\frac{M^2}{\epsilon}\log^{3/2}\Bigl(\frac{M}{\epsilon}\Bigr)\log\log\Bigl(\frac{M}{\epsilon}\Bigr)\Bigr).
\end{align}
This completes the proof.
\end{proof}

\subsubsection{Quantum MALA for estimating the normalizing constant}
We now describe how to combine quantum annealing with quantum MALA to reduce the query complexity of estimating the normalizing constants.

We begin with the following lemma on non-destructive mean estimation.
\begin{lemma}[Non-destructive mean estimation with quantum MALA]\label{lem:non_destructive_MALA}
  For $\epsilon<1$, given $\widetilde{O}(\epsilon^{-1})$ copies of a state $\ket{\widetilde{\rho}_{i-1}}$ such that $\|\ket{\widetilde{\rho}_{i-1}}-\ket{\rho_{i-1}}\|\leq \epsilon$, there exists a quantum procedure that outputs $\widetilde{g}_i$ such that
  \begin{align}
    \left|\widetilde{g}_i-\E_{\rho_i}[g_i]\right|\leq \epsilon \cdot \E_{\rho_i}[g_i]
  \end{align}
  with success probability $1-o(1)$ using
  \begin{align}
    \widetilde{O}(\kappa^{1/2}d^{1/2}\epsilon^{-1})
  \end{align}
  steps of the quantum walk operator corresponding to the MALA with stationary distribution $\rho_i$, where $\delta$ is the spectral gap of the Markov chain. The quantum  procedure also returns $\widetilde{O}(\epsilon^{-1})$ copies of the state $\ket{\widetilde{\rho}_i}$ such that $\|\ket{\widetilde{\rho}_{i}}-\ket{\rho_{i}}\|\leq \epsilon$.
\end{lemma}

\begin{proof}[Proof sketch.]
  This lemma is a variant of Lemma 4.4 in \citet{chakrabarti2018quantum}. Notice that \cor{MALA_slow_vary} implies that we can prepare $\ket{\widetilde{\rho}_i}$ from $\ket{\widetilde{\rho}_{i-1}}$ using $\widetilde{O}(\kappa^{1/2}d^{1/2}p^{-1})$ quantum walk steps, where $p\leq |\langle \rho_i|\rho_{i-1}\rangle|$. By \lem{slowly_varying}, we have $p=\Omega(1)$. The lemma follows by properly choosing the parameters in Lemma 4.4 in \citet{chakrabarti2018quantum}.
\end{proof}

\begin{theorem}[Quantum speedup using MALA, annealing, and quantum walk]\label{thm:quantum-MALA-walk}
  Let $Z$ be the normalizing constant in \prb{log-concave}. Assume we are given the access to query the quantum gradient oracle \eqn{oracle-gradient}. Then there is a quantum algorithm which outputs an estimate $\widetilde Z$, such that
  \begin{align}
  \mathrm{Pr}[(1-\epsilon)Z \le \widetilde Z \le (1+\epsilon)Z] \ge \frac{3}{4},
  \end{align}
  using
  \begin{align}
  \widetilde O\left(d^{3/2}\kappa^{1/2}\epsilon^{-1}\right)
  \end{align}
  queries to the quantum gradient oracle in total.
\end{theorem}
\begin{proof}
  The number of annealing stages is $M=\wt{O}(\sqrt{d})$. At the $i$th stage, we estimate $\E_{\rho_i}[g_i]$ with relative error $\epsilon/M$. Hence, we can apply \lem{non_destructive_MALA} $M$ times, where each application takes
  \begin{align}
    \widetilde{O}(\kappa^{1/2}d^{1/2}(\epsilon/M)^{-1}) = \widetilde{O}(\kappa^{1/2}d\epsilon^{-1})
  \end{align}
MALA quantum walk steps. This process takes
  \begin{align}
    M\cdot \widetilde{O}(\kappa^{1/2}d\epsilon^{-1}) = \widetilde{O}(\kappa^{1/2}d^{3/2}\epsilon^{-1})
  \end{align}
  steps in total.

  By \lem{continuous_walk}, each step of the quantum walk operator can be implemented by querying the gradient oracle and the evaluation oracle $O(1)$ times. Therefore, we can estimate $Z$ with relative error $\epsilon$ using $\widetilde{O}(\kappa^{1/2}d^{3/2}\epsilon^{-1})$ queries to the gradient and evaluation oracles.
\end{proof}

\subsection{Quantum multilevel Langevin algorithms}\label{append:quantum-multilevel}

We now consider an alternative approach for estimating the normalizing constant, by replacing MALA by a multilevel Langevin approach. More concretely, for each sample we perform the underdamped Langevin diffusion (ULD) or the randomized midpoint method for underdamped Langevin diffusion (ULD-RMM) that has an improved dependence on the dimension, and apply the multilevel Monte Carlo (MLMC) to preserve the dependence on the accuracy.

Multilevel Monte Carlo methods have attracted extensive attention in stochastic simulations and financial models~\citep{Gil08,Gil15}. This approach was originally developed by~\citet{Hei01} for parametric integration, and used to simulate SDEs in \citet{Gil08}. Considering a general random variable $P$, MLMC gives a sequence of estimators $P_0, P_l, \ldots, P_L$ for approximating $P$ with increasing accuracy and cost, and uses the telescoping sum of $\E[P_l-P_{l-1}]$ to estimate $\E[P]$. For $P_l-P_{l-1}$ with smaller variance but larger cost, MLMC performs fewer samples to reach a given error tolerance, reducing the overall complexity. MLMC has been widely discussed and improved under many settings, and has been used in various applications~\citet{Gil15}.

To estimate normalizing constants, a variant of MLMC has been proposed by \citet[Lemmas C.1 and C.2]{ge2020estimating}. Unlike standard MLMC for bounding the mean-squared error, this approach upper bounds the bias and the variance separately, making the analysis more technically involved. The first quantum algorithm based on MLMC was developed by \citet[Theorem 2]{ALL20}. They upper bound the additive error with high probability (as is common for quantum algorithms). They also observe that the mean-squared error can control both the bias and the variance \citet[Section 2.2]{ALL20} and that the mean-squared error is almost equivalent to the additive error with high probability \citet[Appendix A]{ALL20}. Considering this, we still use the additive error scenario for estimating normalizing constants, both for convenience and for compatibility with the quantum annealing speedup of \thm{quantum-annealing}.

We first introduce the general quantum speedup of MLMC as described in~\citet{ALL20}, and then apply these results to our problem.

\subsubsection{Quantum-accelerated multilevel Monte Carlo method}

We begin by describing the following general result on quantum-accelerated multilevel Monte Carlo (QA-MLMC).

\begin{lemma}[{Theorem 2 of \citet{ALL20}}]\label{lem:QMLMC}
Let $P$ denote a random variable, and let $P_l$ (for $l\in\{0,1,\ldots,L\})$ denote a sequence of random variables such that
$P_l$ approximates $P$ at level $l$. Also define $P_{-1}=0$.
Let $C_l$ be the cost of sampling from $P_l$, and let $V_l$ be the variance of $P_l-P_{l-1}$.
If there exist positive constants $\alpha,
\beta = 2 \hat{\beta},\gamma$
such that
$\alpha\geq \min(\hat{\beta},\gamma)$
and
\begin{itemize}[nosep]
\item $|\E[P_l-P]| = O(2^{-\alpha l})$,
\item { $V_l= O(2^{-\beta l}) = O( 2^{-2\hat{\beta} l})$}, and
\item $C_l = O(2^{\gamma l})$,
\end{itemize}
then for any $\epsilon< 1/e$
there is a quantum algorithm that
estimates $\E[P]$ up to additive error
$\epsilon$ with probability at least 0.99, and with cost
\begin{align}
\begin{cases}
O\Bigl(\epsilon^{-1} (\log 1/\epsilon)^{3/2} (\log\log 1/\epsilon)^2\Bigr), & \hat{\beta}>\gamma, \\ \vspace{.1cm}
O\Bigl(\epsilon^{-1}(\log 1/\epsilon)^{7/2} (\log\log  1/\epsilon)^2\Bigr)
, & \hat{\beta}=\gamma, \\
O\Bigl(\epsilon^{-1-(\gamma-\hat{\beta})/\alpha} (\log 1/\epsilon)^{3/2} (\log\log 1/\epsilon)^2\Bigr), & \hat{\beta}<\gamma.
\end{cases}
\end{align}
\end{lemma}

We apply this result to the payoff model of general stochastic differential equations (SDEs) as discussed in~\citet{ALL20}. Consider an SDE
\begin{align}\label{eqn:SDE}
\d{X_t} = \mu(X_t,t)\,\d t + \sigma(X_t,t)\,\d W_t
\end{align}
for $t\in[0,T]$, where we assume $\mu$ and $\sigma$ are Lipschitz continuous. Given an initial condition $X_0$ and an evolution time $T>0$, we aim to compute
\begin{align}
\E[\P(X_T)],
\end{align}
where $\P(X)$ is the so-called payoff function as a functional of $X$. In \lem{QMLMC}, we denote $\P(X_T)$ as $P$, and the goal is to estimate $\E[P]$.

We also consider a numerical scheme that produces $\widehat{X}_k$ with time step size $h=T/n$. We say the scheme is of strong order $r$ if for any $m \in \{1,2\}$, there exists a constant $C_m$ such that
\begin{align}\label{eqn:def_strong_r}
\E\left(\|\widehat{X}_n-X_T\|^m\right) \leq C_mh^{rm}.
\end{align}
Note that it suffices to verify this condition for $m=2$ since $(\E\|\widehat{X}_n-X_T\|)^2\le\E\|\widehat{X}_n-X_T\|^2$. We further assume the coefficients of the scheme are Lipschitz continuous. For the discretization $n_l=2^l$ with step size $h=T/2^l$, we let $P_l$ denote $\P(\widehat{X}_{n_l})$, an estimator of $P$.

Finally, we assume $\P(X)$ is $L_P$-Lipschitz continuous. Thus, we have satisfied the three assumptions of Proposition 2 of \citet{ALL20}, which estimates the rates of $|\E[P_l-P]|, V_l, C_l$ in \lem{QMLMC}. We state a simpler version as follows.

\begin{lemma}[{Proposition 2 of \citet{ALL20}}]\label{lem:alpha-beta-gamma}
Given an SDE and a scheme of strong order $r$ with Lipschitz continuous constants, and given a Lipschitz continuous payoff function $\P$, we have $\alpha = r$, $\beta = 2r$, and $\gamma = 1$.
\end{lemma}

Note that while we relax the definition of a scheme of strong order $r$, our definition \eqn{def_strong_r} is sufficient to prove \lem{alpha-beta-gamma}. More concretely, in the proof of Proposition 2 of \citet{ALL20}, we have the following simplified inequalities:
\begin{align}
|\E[P_l-P]| \le \E|\P(\widehat{X}_n)-\P(X_T)| \le L_P\E\|\widehat{X}_n-X_T\| \le L_PC_1h^{r} = O(2^{-rl}),
\end{align}
\begin{align}
V_l \le \E|\P(\widehat{X}_n)-\P(X_T)|^2 \le L_P\E\|\widehat{X}_n-X_T\|^2 \le L_PC_2h^{2r} = O(2^{-2rl}),
\end{align}
\begin{align}
C_l = O(n_l) = O(2^l),
\end{align}
and therefore $\alpha = r$, $\beta = 2r$, and $\gamma = 1$.

Finally, we can characterize the performance of QA-MLMC as follows.

\begin{lemma}[{Theorem 3 of \citet{ALL20}}]\label{lem:QMLMC-SDE}
Given an SDE and a scheme of strong order $r$ with Lipschitz continuous constants, and given a Lipschitz continuous payoff function $\P$, QA-MLMC estimates $\E[\P]$ up to additive error $\epsilon$ with probability at least 0.99 with cost
\begin{align}
O\Bigl(\epsilon^{-1} (\log 1/\epsilon)^{3/2} (\log\log 1/\epsilon)^2\Bigr)&, \quad r > 1, \\
O\Bigl(\epsilon^{-1}(\log 1/\epsilon)^{7/2} (\log\log 1/\epsilon)^2\Bigr)&, \quad r = 1, \\
O\Bigl(\epsilon^{-1/r} (\log 1/\epsilon)^{3/2} (\log\log 1/\epsilon)^2\Bigr)&, \quad  r < 1.
\end{align}
\end{lemma}

Note that we can amplify the success probability to $1-\delta$ for arbitrarily small $\delta>0$ using the powering lemma (\lem{powering-lemma}).

\subsubsection{Quantum-accelerated multilevel Langevin method}

We have described ULD and ULD-RMM in \algo{ULD} and \algo{ULD-RMM}, respectively. We now apply these schemes to simulate the underdamped Langevin dynamics as the SDE. According to \eqn{def_strong_r}, ULD and ULD-RMM are schemes of order 1 and 1.5, respectively.

Let the payoff function $\P$ be $g_i$ as defined in \eqn{distribution-g}. Our goal is to estimate the mean of $\P(\widehat{X}_{n_l}) = g_i(\widehat{X}_{n_l})$ using several samples $\widehat{X}_{n_l}$ produced by ULD or ULD-RMM. If $g$ is assumed to be $L_g$-Lipschitz as in Lemma C.2 of \citet{ge2020estimating}, we have a Lipschitz continuous payoff function $\P$ with $L_P=L_g$. Although $g_i = \exp\bigl(\frac{\|x\|^2}{\sigma_i^2(1+\alpha^{-1})}\bigr)$ is not Lipschitz, according to Section 4.3 of \citet{ge2020estimating}, we can truncate at large $x$ and replace $g_i$ by $h_i := \min \bigl\{g_i , \exp\bigl(\frac{(r^+_i)^2}{\sigma_i^2(1+\alpha^{-1})}\bigr)\bigr\}$ with
\begin{align}\label{eqn:truncation}
\alpha &= \widetilde O\biggl(\frac{1}{\sqrt{d}\log(1/\epsilon)}\biggr), \\
r_i^+ &= \E_{\rho_{i+1}}\|x\| + \Theta(\sigma_i\sqrt{(1+\alpha)\log(1/\epsilon)}),
\end{align}
to ensure $\frac{h_i}{\E_{\rho_i}g_i}$ is $O(\sigma_i^{-1})$ Lipschitz and $\bigl|\E_{\rho_i}(h_i-g_i)\bigr|<\epsilon$ by Lemmas C.7 and C.8 of \citet{ge2020estimating}. For simplicity, as in Section 4.2 of \citet{ge2020estimating}, we regard $g_i$ as a Lipschitz continuous payoff function in our main results.

Thus, using \lem{QMLMC-SDE} to estimate $\frac{Z_{i+1}}{Z_{i}} = \E_{\rho_i}(g_i)$,
 QA-MLMC using either ULD or ULD-RMM can reduce the $\epsilon$-dependence of the number of steps to $\widetilde O(\epsilon^{-1})$. Then each step of ULD or ULD-RMM uses the value of $\nabla f(x)$ about $O(1)$ times as shown in \algo{ULD} and \algo{ULD-RMM}.

Having described the implementations of quantum inexact ULD and ULD-RMM, we now state the quantum speedup for estimating normalizing constants using multilevel ULD and annealing, or multilevel ULD-RMM and annealing, as follows.

\begin{theorem}[Quantum speedup using multilevel ULD and annealing]\label{thm:quantum-annealing-multilevel-ULD}
Let $Z$ be the normalizing constant in \prb{log-concave}. Assume we are given access to the quantum gradient oracle \eqn{oracle-gradient}. Then there is a quantum algorithm which outputs an estimate $\widetilde Z$ such that
\begin{align}\label{eqn:quantum-annealing-multilevel-ULD-prob}
\mathrm{Pr}[(1-\epsilon)Z \le \widetilde Z \le (1+\epsilon)Z] \ge \frac{3}{4}
\end{align}
using
\begin{align}\label{eqn:quantum-annealing-multilevel-ULD-query}
\widetilde O\Big(\frac{d^{3/2}\kappa^2}{\epsilon}\Big)
\end{align}
queries to the quantum gradient oracle.
\end{theorem}

\begin{proof}
As in \thm{quantum-annealing} and \thm{quantum-MALA-walk}, for $i\in\range{M}$ with $M$ stages, we estimate $\E_{\rho_i}(g_i)$ with output $\widetilde g_i$ up to additive error $(\epsilon/2M)\E_{\rho_i}(g_i)$ with failure probability $1/4M$, which ensures $|\widetilde Z-Z| \le \epsilon Z$ with probability at least $\frac{3}{4}$.

According to \lem{QMLMC-SDE}, each sample of $\rho_i$ using multilevel ULD uses $\widetilde O(\frac{M\kappa^2\sqrt{d}}{\epsilon})$ queries to the quantum evaluation oracle \eqn{oracle-evaluation} used in \algo{quantum-IULD} or \algo{quantum-IULD-RMM}. For $M=\widetilde O(\sqrt{d})$ stages, the query complexity of estimating the normalizing constant is $\widetilde O(\frac{M^2\kappa^2\sqrt{d}}{\epsilon})=\widetilde O(\frac{d^{3/2}\kappa^2}{\epsilon})$.
\end{proof}

\begin{theorem}[Quantum speedup using multilevel ULD-RMM and annealing]\label{thm:quantum-annealing-multilevel-ULD-RMM}
Let $Z$ be the normalizing constant in \prb{log-concave}. Assume we are given the access to the quantum gradient oracle \eqn{oracle-gradient}. Then there is a quantum algorithm which outputs an estimate $\widetilde Z$ such that
\begin{align}\label{eqn:quantum-annealing-multilevel-ULD-RMM-prob}
\mathrm{Pr}[(1-\epsilon)Z \le \widetilde Z \le (1+\epsilon)Z] \ge \frac{3}{4}
\end{align}
using
\begin{align}\label{eqn:quantum-annealing-multilevel-ULD-RMM-query}
\widetilde O\Big(\frac{d^{7/6}\kappa^{7/6}+d^{4/3}\kappa}{\epsilon}\Big)
\end{align}
queries to the quantum gradient oracle.
\end{theorem}

\begin{proof}
As above, for $i\in\range{M}$ with $M$ stages, we estimate $\E_{\rho_i}(g_i)$ with output $\widetilde g_i$ up to additive error $(\epsilon/2M)\E_{\rho_i}(g_i)$ with failure probability $1/4M$.

According to \lem{QMLMC-SDE} and \lem{ULD-RMM}, each sample of $\rho_i$ using multilevel ULD uses $\widetilde O(\frac{M(\kappa^{7/6}d^{1/6}+\kappa d^{1/3})}{\epsilon})$ queries to the quantum gradient oracle \eqn{oracle-gradient} or the quantum evaluation oracle \eqn{oracle-evaluation} (with additional $\widetilde O(1)$ cost). For $M=\widetilde O(\sqrt{d})$ stages, the query complexity of estimating the normalizing constant is $\widetilde O(\frac{M^2(\kappa^{7/6}d^{1/6}+\kappa d^{1/3})}{\epsilon})=\widetilde O(\frac{d^{7/6}\kappa^{7/6}+d^{4/3}\kappa}{\epsilon})$.
\end{proof}

\section{Proof of the Quantum Lower Bound}\label{append:quantum-lower}
To prove \thm{lower-bound}, we use the following quantum query lower bound on the Hamming weight problem.
\begin{proposition}[{Theorem 1.3 of \citet{nayak1999quantum}}]\label{prop:Hamming-weight}
For $x=(x_1,\ldots,x_n)\in\{0,1\}^{n}$, let $\|x\|_{1}=\sum_{i=1}^{n}x_{i}$ be the Hamming weight of $x$. Furthermore, let $\ell,\ell'$ be integers such that $0\leq\ell<\ell'\leq n$. Define the partial boolean function $f_{\ell,\ell'}$ on $\{0,1\}^{n}$ as
\begin{align}
f_{\ell,\ell'}(x)=\begin{cases}
    0 & \text{if } \|x\|_{1}=\ell \\
    1 & \text{if } \|x\|_{1}=\ell'.
  \end{cases}
\end{align}
Let $m\in\{\ell,\ell'\}$ be such that $|\frac{n}{2}-m|$ is maximized, and let $\Delta=\ell'-\ell$.
Then given the quantum query oracle
\begin{align}
O_{x}|i\>|b\>=|i,b\oplus x_i\>\quad\forall i\in[n], b\in\{0,1\},
\end{align}
the quantum query complexity of computing the function $f_{\ell,\ell'}$ is $\Theta(\sqrt{n/\Delta}+\sqrt{m(n-m)}/\Delta)$.
\end{proposition}

Now we prove \thm{lower-bound} using a construction motivated by Section 5 of \citet{ge2020estimating}.

\begin{proof}%
We start from a basic function $f_{0}(x)=\frac{\|x\|^{2}}{2}$. The partition function of $f_{0}$ is
\begin{align}
\int_{\R^{k}}e^{-f_{0}(x)}\d x=(2\pi)^{k/2}.
\end{align}
We then construct $n$ cells in $\R^{d}$. Without loss of generality we assume that $n^{1/k}$ is an integer, and let $l:=1/(\sqrt{k}n^{1/k})$. We partition $[-1/\sqrt{k},1/\sqrt{k}]$ into $n^{1/k}$ intervals, each having length $2l$. Let $I_{i}$ denote the $i^{\text{th}}$ interval, where $i\in[n^{1/k}]$. Each cell will thus be represented as a $k$-tuple $(i_{1},\ldots,i_{k})\in\{1,2,\ldots,n^{1/k}\}^{k}$ corresponding to $I_{i_{1}}\times\cdots\times I_{i_{k}}\subset \R^{k}$.

Next, each cell $\tau=(i_{1},\ldots,i_{k})$ with center denoted $v_{\tau}$ is assigned one of two types (as detailed below), and we let
\begin{align}
f(x)=\begin{cases}
    f_{0}(x) & \text{if cell $\tau$ is of type 1}  \\
    f_{0}(x)+c_{\tau}q(\frac{1}{l}(x-v_{\tau})) & \text{if cell $\tau$ is of type 2}.
  \end{cases}
\end{align}
The function $q$ and the normalizing constant $c_{\tau}$ are carefully chosen, following Lemma D.1 in \citet{ge2020estimating}, such that
\begin{itemize}
\item $f(x)$ is 1.5-smooth and 0.5-strongly convex; and
\item the partition function $Z_{f}=\int_{\R^{k}}e^{-f(x)}\d x=(2\pi)^{k/2}-C\cdot\frac{n_{2}}{n}$, where $n_{2}$ is the number of type-2 cells, and $C$ is at least $\Omega(l^{2})$.
\end{itemize}

With these properties, we consider two functions as follows. We choose $\delta$ such that $\epsilon=\Theta(\delta^{1+4/k})$. One of the functions has a $1/2+\delta$ fraction of its cells of type 1 (and a $1/2-\delta$ fraction of type 2). The other function has a $1/2-\delta$ fraction of its cells of type 1 (and a $1/2+\delta$ fraction of type 2). Note that one query to the quantum evaluation oracle \eqn{oracle-evaluation} can be implemented using one quantum query to the binary information indicating the type of the corresponding cell. In addition, by \prop{Hamming-weight} with $\ell=(1-\delta)n/2$ and $\ell'=(1+\delta)n/2$, it takes $\Omega(1/\delta)$ quantum queries to distinguish whether there are $(1+\delta)n/2$ or $(1-\delta)n/2$ cells of type 1. Since $C=\Omega(l^{2})$, the partition functions of the two functions differ by a multiplicative factor of at least $1+\Omega(l^{2}\delta)$, where $l=\Theta(n^{-1/k})=\Theta(\delta^{2/k})$, and hence $l^{2}\delta=\Theta(\delta^{1+4/k})=\Theta(\epsilon)$. The quantum query complexity is therefore $\Omega(1/\delta)=\Omega(\epsilon^{-\frac{1}{1+4/k}})$ as claimed.
\end{proof}

\end{document}